\documentclass{llncs} 

\usepackage[T1]{fontenc}

\usepackage[dvips]{graphicx}
\usepackage{amsmath}
\usepackage{amssymb}
\usepackage{color}
\usepackage{enumitem}
\usepackage{subfig}
\usepackage{color}
\usepackage{multicol}
\usepackage{pgf,tikz}
\usetikzlibrary{arrows}
\usetikzlibrary{patterns}

\newtheorem{obs}{Observation}
\usepackage{algpseudocode}

\usepackage{algorithm}

\usepackage{epsfig}
\usepackage{tikz}

\title{Covering the Boundary of a Simple Polygon \\ with Geodesic Unit Disks}
\author{George Rabanca \and Ivo Vigan \thanks{Research supported by NSF grant 1017539}}
\institute{Department of Computer Science, City University of New York, \\The Graduate Center, New York, NY 10016, USA.}

\begin{document}

\maketitle


\begin{abstract}          
We consider the problem of covering the boundary of a simple polygon on $n$ vertices using the minimum number of geodesic unit disks. We present an $O(n \log^2 n+k)$ time $2$-approximation algorithm for finding the centers of the disks, with $k$ denoting the number of centers found by the algorithm.

\end{abstract}

\section{Introduction and Main Results}




For two points $u$ and $v$ in a simple polygon $P$, the \emph{geodesic distance}, denoted by $d(u,v)$, is the length of the shortest path between $u$ and $v$ inside $P$. A \emph{geodesic unit disk} $D(v)$ centered at a point $v \in P$ is the set of all points in $P$ whose geodesic distance to $v$ is at most $1$.

The \emph{boundary} of $D(v)$, denoted by $\partial D(v)$, consists of all points in $P$ which are either exactly at distance $1$ from $v$ or at distance at most $1$ from $v$ but contained in the polygon boundary $\partial P$. The \emph{interior} of $D(v)$, denoted by $int(D(v))$,  consists of all the points of $D(v)$ not contained on the boundary of $D(v)$, i.e.,  $int(D(v)) = D(v) \setminus \partial D(v)$,  as shown in Fig.~\ref{geoDisk}. 

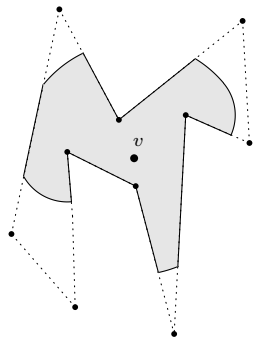
\begin{figure}
\center
\definecolor{uuuuuu}{rgb}{0.27,0.27,0.27}
\definecolor{qqqqff}{rgb}{0.9,0.9,.9}

\definecolor{dddddd}{rgb}{0.7,0.6,0.4}

\begin{tikzpicture}[line cap=round,line join=round,>=triangle 45,x=.7cm,y=.7cm]
\clip(1.5,-0.5) rectangle (6.4,6.1);
\fill[color=qqqqff,fill=qqqqff,fill opacity=1.0] (2.73,3.12) -- (1.9,2.66) -- (2.01,3.18) -- cycle;
\fill[color=qqqqff,fill=qqqqff,fill opacity=1.0] (4,3) -- (2.01,3.18) -- (2.27,4.4) -- cycle;
\fill[color=qqqqff,fill=qqqqff,fill opacity=1.0] (4,3) -- (4.03,2.47) -- (2.73,3.12) -- cycle;
\fill[color=qqqqff,fill=qqqqff,fill opacity=1.0] (4,3) -- (3.71,3.73) -- (3.03,5) -- cycle;
\fill[color=qqqqff,fill=qqqqff,fill opacity=1.0] (4,3) -- (4.46,0.83) -- (4.03,2.47) -- cycle;
\fill[color=qqqqff,fill=qqqqff,fill opacity=1.0] (4,3) -- (4.83,0.94) -- (4.98,3.82) -- cycle;
\fill[color=qqqqff,fill=qqqqff,fill opacity=1.0] (4,3) -- (5.16,4.89) -- (3.71,3.73) -- cycle;

\draw [shift={(4,3)},color=qqqqff,fill=qqqqff,fill opacity=1.0]  (0,0) --  plot[domain=4.92:5.1,variable=\t]({1*2.22*cos(\t r)+0*2.22*sin(\t r)},{0*2.22*cos(\t r)+1*2.22*sin(\t r)}) -- cycle ;
\draw [shift={(4,3)},color=qqqqff,fill=qqqqff,fill opacity=1.0]  (0,0) --  plot[domain=0.7:1.02,variable=\t]({1*2.22*cos(\t r)+0*2.22*sin(\t r)},{0*2.22*cos(\t r)+1*2.22*sin(\t r)}) -- cycle ;
\draw [shift={(4.98,3.82)},color=qqqqff,fill=qqqqff,fill opacity=1.0]  (0,0) --  plot[domain=-0.42:0.7,variable=\t]({1*0.94*cos(\t r)+0*0.94*sin(\t r)},{0*0.94*cos(\t r)+1*0.94*sin(\t r)}) -- cycle ;

\draw [shift={(4,3)},color=qqqqff,fill=qqqqff,fill opacity=1.0]  (0,0) --  plot[domain=2.02:2.46,variable=\t]({1*2.22*cos(\t r)+0*2.22*sin(\t r)},{0*2.22*cos(\t r)+1*2.22*sin(\t r)}) -- cycle ;

\draw [shift={(4.98,3.82)}] plot[domain=-0.42:0.7,variable=\t]({1*0.94*cos(\t r)+0*0.94*sin(\t r)},{0*0.94*cos(\t r)+1*0.94*sin(\t r)});
\draw [shift={(4,3)}] plot[domain=4.92:5.1,variable=\t]({1*2.22*cos(\t r)+0*2.22*sin(\t r)},{0*2.22*cos(\t r)+1*2.22*sin(\t r)});
\draw [shift={(4,3)}] plot[domain=2.02:2.46,variable=\t]({1*2.22*cos(\t r)+0*2.22*sin(\t r)},{0*2.22*cos(\t r)+1*2.22*sin(\t r)});
\draw [shift={(4,3)}] plot[domain=0.7:1.02,variable=\t]({1*2.22*cos(\t r)+0*2.22*sin(\t r)},{0*2.22*cos(\t r)+1*2.22*sin(\t r)});

\draw [shift={(2.73,3.12)},color=qqqqff,fill=qqqqff,fill opacity=1.0]  (0,0) --  plot[domain=3.65:4.79,variable=\t]({1*0.95*cos(\t r)+0*0.95*sin(\t r)},{0*0.95*cos(\t r)+1*0.95*sin(\t r)}) -- cycle ;

\draw [shift={(2.73,3.12)}] plot[domain=3.65:4.79,variable=\t]({1*0.95*cos(\t r)+0*0.95*sin(\t r)},{0*0.95*cos(\t r)+1*0.95*sin(\t r)});

\draw[dotted] (2.88,0.17)-- (1.67,1.56);
\draw[dotted] (1.67,1.56)-- (2.58,5.83);
\draw[dotted] (2.58,5.83)-- (3.71,3.73);
\draw[dotted] (3.71,3.73)-- (6.06,5.61);
\draw[dotted] (6.06,5.61)-- (6.19,3.29);
\draw[dotted] (6.19,3.29)-- (4.98,3.82);
\draw[dotted] (4.98,3.82)-- (4.76,-0.34);
\draw[dotted] (4.76,-0.34)-- (4.03,2.47);
\draw[dotted] (4.03,2.47)-- (2.73,3.12);
\draw[dotted] (2.73,3.12)-- (2.83,1.9);
\draw[dotted] (2.83,1.9)-- (2.88,0.17);

\draw [color=qqqqff] (2.73,3.12)-- (1.9,2.66);
\draw [color=black] (1.9,2.66)-- (2.01,3.18);
\draw [color=qqqqff] (2.01,3.18)-- (2.73,3.12);
\draw [color=qqqqff] (4,3)-- (2.01,3.18);
\draw [color=black] (2.01,3.18)-- (2.27,4.4);
\draw [color=qqqqff] (2.27,4.4)-- (4,3);
\draw [color=qqqqff] (4,3)-- (4.03,2.47);
\draw [color=black] (4.03,2.47)-- (2.73,3.12);
\draw [color=qqqqff] (2.73,3.12)-- (4,3);
\draw [color=qqqqff] (4,3)-- (3.71,3.73);

\draw [color=qqqqff] (3.03,5)-- (4,3);
\draw [color=qqqqff] (4,3)-- (4.46,0.83);
\draw [color=black] (4.46,0.83)-- (4.03,2.47);
\draw [color=qqqqff] (4.03,2.47)-- (4,3);
\draw [color=qqqqff] (4,3)-- (4.83,0.94);
\draw [color=black] (4.83,0.94)-- (4.98,3.82);
\draw [color=qqqqff] (4.98,3.82)-- (4,3);
\draw [color=qqqqff] (4,3)-- (5.16,4.89);
\draw [color=black] (5.16,4.89)-- (3.71,3.73);
\draw [color=qqqqff] (3.71,3.73)-- (4,3);
\draw [color=black] (4.98,3.82) -- (5.84,3.44);
\draw [color=black](2.73,3.12)  -- (2.81,2.17) ;
\draw [color=black] (3.71,3.73)-- (3.03,5);

\begin{scriptsize}
\fill [color=black] (2.88,0.17) circle (1.1pt);
\fill [color=black] (1.67,1.56) circle (1.1pt);
\fill [color=black] (2.58,5.83) circle (1.1pt);
\fill [color=black] (3.71,3.73) circle (1.1pt);
\fill [color=black] (6.06,5.61) circle (1.1pt);
\fill [color=black] (6.19,3.29) circle (1.1pt);
\fill [color=black] (4.98,3.82) circle (1.1pt);
\fill [color=black] (4.76,-0.34) circle (1.1pt);
\fill [color=black] (4.03,2.47) circle (1.1pt);
\fill [color=black] (2.73,3.12) circle (1.1pt);
\fill [color=black] (4,3) circle (1.5pt);
\draw[color=black] (4.08,3.3) node {$v$};
\end{scriptsize}
\end{tikzpicture}
\caption{A polygon (dotted) containing a geodesic disk centered at $v$, whose interior is depicted in gray and its boundary is drawn in black.}
\label{geoDisk}
\end{figure}

A collection of geodesic disks \emph{covers} the polygon boundary $\partial P$, if each point of $\partial P$ is contained in at least one disk. In this paper we present an $O(n \log^2 n + k)$ time $2$-approximation algorithm which finds a collection of geodesic unit disks covering the boundary of a simple polygon on $n$ vertices, with $k$ denoting the number of disks found by the algorithm. The algorithm then returns the centers of the disks. {We consider the setting where the centers can be placed anywhere inside the polygon, but the algorithm can be easily modified to restrict the centers to lie on $\partial P$.}  Furthermore, the \emph{number} of disks can be computed in time $O(n \log^2 n)$.

While it follows from Theorem 7 of \cite{viganPack} that our problem is $\mathsf{NP}$-hard in polygons with holes, its complexity remains open in simple polygons.

The main motivation for studying this problem comes from sensor networks, where {\emph{Barrier Coverage} problems have been studied extensively} (see for example \cite{conf/algosensors/BeregK09},\cite{5210107},\cite{Chen:2008:MGQ:1374618.1374674},\cite{Kumar:20059},\cite{Liu:2008:SBC:1374618.1374673},\cite{conf/infocom/SaipullaWLW09},\cite{4520201}). In a Barrier Coverage problem the goal is to place few sensors or guards to detect any intruder into a given region. {The algorithm in this paper can be applied to this context}: given a region, bounded by a piecewise linear closed border, such as a fence, place few guards inside the fenced region, such that wherever an intruder cuts through the fence, the closest guard is at most distance one away. Another way of looking at this problem is from an Art Gallery perspective (see for example \cite{ORourke:1987:AGT:40599}), where the polygon represents a gallery and, regardless where on the wall a painting is hanged, the closest guard is at most a distance one away.

\subsection{Related Work}
Several papers (\cite{Funke:2007},\cite{Huang:2003},\cite{1939},\cite{ijdsn/Ko12},\cite{936985},\cite{6226360}) study full coverage of geometric regions with Euclidean disks. For an overview of optimal coverings of squares and triangles with disks see Chapter 1.7  of \cite{unsolv2}. 

In the context of Barrier Coverage, \cite{Cabello:2013:CSP:2462356.2462383} presents a polynomial time algorithm which for two points in the plane and a set of Euclidean disks selects a minimal subset of the disks which separates the two points. Extending the problem to $k$ points, an $O(1)$-approximation algorithm was presented in \cite{Gibson:2011:IPU:2040572.2040580} and $\mathsf{NP}$-hardness was shown in \cite{DBLP:journals/corr/abs-1303-2779}. The same two point separation problem was studied in \cite{minCellCon} when segments instead of disks are given.

{Covering a simple polygon with a single geodesic disk of minimum radius has been studied in \cite{pollackGeod} and a linear-time algorithm is presented in \cite{heeKapLinear}. {An output sensitive algorithm for computing geodesic disks for a given set of centers and a fixed radius is presented in \cite{geodPac}}. 

\subsection{Paper Organization}
This paper is organized as follows. In Section \ref{algoSec} we present the algorithm and show that it runs in time $O(n \log^2 n + k)$. In Section \ref{approx} we  {prove that the number of centers placed by the algorithm is at most twice the minimum number of centers needed to cover the polygon boundary}.  In Section \ref{refAna} we show that a simple linear time algorithm achieves an asymptotically optimal approximation ratio when the polygon perimeter is much larger than $n$. All the missing proofs can be found in the Appendix.
\section{The Algorithm and Its Running Time}
\label{algoSec}

Our algorithm makes use of several properties of geodesic Voronoi diagrams which we review below.

\subsection{Geodesic Voronoi diagrams}
\label{vorDia}

A \emph{furthest-site geodesic Voronoi diagram} of $k$ sites in a simple polygon $P$ on $n$ vertices is a decomposition of $P$ into cells such that all points in a cell have the same site \emph{furthest} away from them {(in the geodesic metric)}. As shown in \cite{BorisFurthest}, it has combinatorial complexity $O(n+k)$ and can be constructed in time $O((n+k)\log(n+k))$. In Section 2.8 of \cite{BorisFurthest} it is shown that these combinatorial and time complexities are with respect to a \emph{refinement} (also called a shortest path partition) of the Voronoi edges. For all points on a refined edge it holds that their shortest paths to each of the two furthest sites are combinatorially equivalent, i.e., they consist of the same sequence of polygon vertices respectively. Furthermore, Section 3.3 of \cite{BorisFurthest} defines for each of the $O(n+k)$ refined edges, and for each of the two furthest sites $s_1$ and $s_2$ defining a Voronoi edge $e$, the \emph{anchor points} $a_e(s_1)$, $a_e(s_2)$ which are the last points on the shortest path from $s_1$, $s_2$ respectively to any point on $e$. Those anchors can be computed in total $O(n+k)$ time and each time we compute a   furthest-site geodesic Voronoi diagram we store the anchors as well as the distance to its site at the refined Voronoi edges. An additional property of this Voronoi diagram is that its edges form a tree, rooted at the \emph{geodesic center} of the $k$ sites, which is defined as the point that minimizes the maximum distance to any of the sites (see Corollary 2.9.3 of \cite{BorisFurthest}). Therefore, the geodesic center of the sites can be obtained within the same time bound.

The second data structure we use is the \emph{closest-site geodesic Voronoi diagram}  { which, for $k$ sites in a simple polygon $P$ on $n$ vertices,}
is a decomposition of $P$ into cells such that all points in a cell have the same site \emph{closest} to them {(in the geodesic metric)}. It has combinatorial complexity $O(n+k)$ and it can be constructed in time $O((n+k)\log(n+k))$ (see \cite{borisClose}). 

\subsection{The \textsc{ContiguousGreedy} Algorithm}
In this section we describe a greedy $2$-approximation algorithm which finds a collection of geodesic unit disks which cover the boundary of $P$ and returns the set of disk centers. It starts at vertex $v_1$ of  the vertices $v_1, \ldots, v_n$ of $P$ and {iteratively extends a contiguous cover $\Gamma$ of $\partial P$  (in clockwise order) by the maximum amount that can be covered with a single geodesic disk.} We denote the clockwise endpoint of $\Gamma$ by $c$, thus initially $\Gamma = \{v_1\}$ and $c = v_1$. \\

We cover segment portions longer than $2$ in time linear in the minimum number of disks needed to cover them. With $v_u$ denoting the first uncovered vertex in the current iteration, we}  partially cover $\overline{cv_u}$ by adding $\lceil d(c, v_u)/2 \rceil - 1$ {centers sequentially on $\overline{cv_u}$}.  By this, we assure that none of  {those disks} contains $v_u$ and, since each disk contains a boundary portion of length $2$, the disks placed are indeed optimal with respect to the greedy contiguous extension criterion.



\begin{definition}
For a polygonal chain $C$, we denote by $\|C\|$ the sum of the lengths of its line segments and we refer to the number of vertices of $C$ by $|C|$.
\end{definition}

\begin{definition}
For two points $u,v \in \partial P$, we denote the portion of $\partial P$ in clockwise orientation between $u$ and $v$ by $\partial P[u,v]$.
\end{definition}

If $c$ does not lie on a long segment, we compute the next endpoint $c'$ which extends $\Gamma$ in clockwise order by a maximum length boundary portion which can be covered by a single geodesic unit disk. We do this   by finding the first vertex $v_u$ (in clockwise order) such that $\partial P[c, v_u]$ cannot be contained in a single geodesic unit disk. This test is done by calling the \textsc{TestCover}$(c,v)$ procedure discussed below, which, for a boundary point $c$ and a vertex $v$ tests whether $\partial P[c, v]$ can be covered with a single geodesic unit disk.  
 {If $v_i$ is the first vertex in clockwise order after $c$, we find $v_u$ by first using exponential search with the \textsc{TestCover} predicate with $c$  fixed and $v$ set to $v_{i+1}, v_{i+2}, v_{i+4}, ... v_{i+2^k}, ...$ respectively in consecutive steps until \textsc{TestCover} returns $false$ or $i + 2^k > n$.} 
This defines an {index}-interval containing {the index} $u$ which can then be found using a simple binary search. 

After finding $v_u$ and thereby fully determining the sequence of vertices covered in the current iteration, we use the \textsc{AugmentShort} procedure -- discussed below -- to compute the new endpoint $c'$ of $\Gamma$ as well as the center of the next disk. \\

\vspace{-10px}

{\center
\fbox{%
\begin{minipage}{\textwidth}
\noindent {\bf \textsc{ContiguousGreedy}}\\
\noindent $c \gets v_1, \; v_u \gets v_2$\\
\noindent $S \gets \emptyset$\\
\noindent{\bf while} $\partial P$ not covered: 
\begin{itemize}[leftmargin=.5in]
\vspace{-5px}
 \item  [1.]  If $\overline{cv_u}$ is longer than 2 \\
\hspace*{10px}  compute centers on $\overline{cv_u}$ at steps of $2$; add them to $S$; update $c$
\item [2.] Update $v_u$ to the first vertex s.t. $\partial P[c, v_u]$ cannot be covered by a single disk,
		 using Exponential and Binary Search with predicate \textsc{TestCover} 
\item [3.]  Use \textsc{AugmentShort} to cover the vertices between $c$ and $v_u$, and a maximal portion of the edge $\overline{v_{u-1}v_u}$; add new center to $S$ and update $c$
\end{itemize}
\vspace{-5px}
\noindent{\bf end while}\\
\noindent{\bf return} $S$
\end{minipage}}
}

\begin{definition}
For two points $u,v$ in a simple polygon $P$, we denote the shortest path in $P$ between $u$ and $v$ by $\pi(u,v)$. We denote the number of its vertices by $|\pi(u,v)|$.
\end{definition}

\begin{definition}[\cite{Toussaint89computinggeodesic}]
A set $Q$ inside a simple polygon $P$ is called \emph{geodesic convex}, if for any two points $u,v \in Q$, the shortest path $\pi(u,v)$ is contained in $Q$. 
\end{definition}



\paragraph{\textbf{\textsc{TestCover}$(c, v)$.}}
This procedure tests for a boundary point $c$ and a polygon vertex $v$ whether $\partial P[c,v]$ can be covered with a single geodesic unit disk. Observe that if a geodesic unit disk {can cover} a set of points, then a geodesic unit disk centered at the \emph{geodesic center} of those points obviously  also covers them. Let $U = U(c,v)$ denote the sequence of point $c$ and all polygon vertices up to (and including) $v$ in clockwise order. \textsc{TestCover} computes the geodesic center of $U$ and returns true iff it has distance at most one to all points in $U$. \\

\noindent {\em Implementation details.} We compute the geodesic center of $U$ in a smaller polygon $Q$ containing $U$. We let $Q$ be $ \partial P[c, v]  \circ \pi(v,c)$, with $\circ$ denoting the concatenation of two polygonal chains sharing two endpoints. Note that $Q$ may have touching sides, but it is not self-intersecting. Such polygons are referred to as \emph{weakly simple} polygons (\cite{Dumitrescu2009112}) and the geodesic distance within them is well defined. Since $Q$ is the concatenation of a boundary part of $P$ and a shortest path in $P$ it follows that $Q$ is geodesic convex in $P$, thus implying that the geodesic center of $U$ in $Q$ is the same point as the geodesic center of $U$ in $P$. We find this geodesic center point by computing the furthest-site geodesic Voronoi diagram $\mathcal{VP}_Q(U)$ of the sites $U$ in $Q$, traversing the (oriented) Voronoi edges to the root and thereby obtain the geodesic center of $U$ (see Section \ref{vorDia}). Then, for each site in $U$ we test whether the distance to the geodesic center is at most one. 

\noindent {\em Computational complexity.}  Computing $\pi(v,c)$ takes time $O(|\pi(v,c)|\log n)$ after $O(n)$ global pre-processing time, using the algorithm of \cite{Guibas1989126}; concatenating two polygonal chains to construct $Q$ takes constant time. Computing $\mathcal{VP}_Q(U)$ takes $O(|Q|\log(|Q|))$ time and the geodesic center can be obtained from $\mathcal{VP}_Q(U)$ in the same time bound. Computing the distance from the geodesic center to all sites in $U$ can be done in time $O(|Q|)$ (see \cite{lineshortpath}), by building the shortest path tree rooted at the geodesic center. 
Therefore, the procedure has an overall time complexity of $O(|Q|\log n)$. \\

Knowing the first vertex $v_u$ such that $\partial P[c,v_u]$ cannot be covered with a single geodesic unit disk, we compute the center of the next disk and compute the new endpoint $c'$ of $\Gamma$ the following \textsc{AugmentShort} procedure.

\paragraph{\textbf{\textsc{AugmentShort}$(c, v_u)$.}}
For the new endpoint $c'$ of $\Gamma$ it needs to hold that $\partial P[c, c']$ can be covered with one geodesic unit disk, and for any $c''  \in \partial P$, with $\|\partial P[c, c'']\| > \|\partial P[c, c']\|$ it is not possible to cover $\partial P[c, c'']$ with a single geodesic unit disk. 
Let $\overline{U} = \overline{U}(c,v_{u-1})$ denote the clockwise sequence of point $c$ and all vertices up to (and including) $v_{u-1}$. We construct $Q = \partial P[c, v_u]  \circ \pi(v_u,c)$ and denote by $A$ the intersection of the geodesic unit disks centered at the points $\overline{U}$ in $Q$. This intersection is non-empty by construction and the center of the next disk lies in $A$. We denote by $I$ the set of all disk-disk intersection points on $\partial A$ as shown in Fig. \ref{algoPic}. { Lemma~\ref{Aint} below justifies the steps taken to find $c'$}.

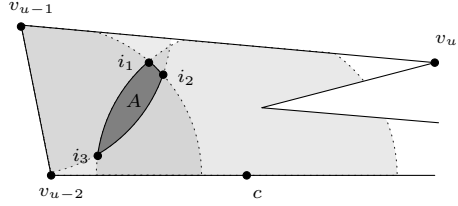
\begin{figure}
\center
\definecolor{uququq}{rgb}{0.25,0.25,0.25}
\definecolor{sqsqsq}{rgb}{0.13,0.13,0.13}
\definecolor{darkgrey}{rgb}{0.5,0.5,0.5}
\begin{tikzpicture}[line cap=round,line join=round,>=triangle 45,x=1.0cm,y=1.0cm]


\draw [shift={(3,1)},dotted,color=sqsqsq,fill=sqsqsq,fill opacity=0.1]  (0,0) --  
plot[domain=0:0.95,variable=\t]({1*2*cos(\t r)+0*2*sin(\t r)},{0*2*cos(\t r)+1*2*sin(\t r)}) --
plot[domain=2.02:3.14,variable=\t]({1*2*cos(\t r)+0*2*sin(\t r)},{0*2*cos(\t r)+1*2*sin(\t r)}) -- cycle ;

\draw [dotted,color=sqsqsq,fill=sqsqsq,fill opacity=0.1]  (0.4,1) --  
plot[domain=0:1.27,variable=\t]({0.4 + 1*2*cos(\t r)+0*2*sin(\t r)},{1 + 0*2*cos(\t r)+1*2*sin(\t r)}) -- (0, 3) -- cycle ;

\draw [shift={(0,3)},dotted,color=uququq,fill=uququq,fill opacity=0.1]  (0,0) --  
plot[domain=-1.37:-0.09,variable=\t]({1*2*cos(\t r)+0*2*sin(\t r)},{0*2*cos(\t r)+1*2*sin(\t r)}) -- cycle ;

\draw [solid,color=black, fill=darkgrey]  
plot[domain=2.28:3.0,variable=\t]({3 + 1*2*cos(\t r)+0*2*sin(\t r)},{ 1+ 0*2*cos(\t r)+1*2*sin(\t r)}) -- 
plot[domain= -1.0:-0.34,variable=\t]({1*2*cos(\t r)+0*2*sin(\t r)},{ 3+ 0*2*cos(\t r)+1*2*sin(\t r)}) -- 
plot[domain= 0.8:0.86,variable=\t]({0.4 + 1*2*cos(\t r)+0*2*sin(\t r)},{ 1+ 0*2*cos(\t r)+1*2*sin(\t r)});

\draw [fill=white, draw=none] (5.5,2.5) -- (3.2, 1.9)-- (5.55, 1.7) -- cycle;
\draw (5.5, 1)-- (3,1)-- (0.4,1)-- (0,3)-- (5.5,2.5) -- (3.2, 1.9)-- (5.55, 1.7);



\begin{scriptsize}
\draw [fill=black] (3,1) circle (1.4pt);
\draw[color=black] (3.14,.75) node {$c$};
\draw [fill=black] (0.4,1) circle (1.4pt);
\draw[color=black] (0.54,.75) node {$v_{u-2}$};
\draw [fill=black] (0,2.98) circle (1.4pt);
\draw[color=black] (0.14,3.22) node {$v_{u-1}$};
\draw [fill=black] (5.5,2.5) circle (1.4pt);
\draw[color=black] (5.66,2.8) node {$v_u$};
\draw[color=black] (1.5,2) node {$A$};

\draw[fill=black] (1.7,2.5) circle (1.4pt);
\draw[color=black] (1.4,2.5) node {$i_1$};
\draw[fill=black] (1.89,2.34) circle (1.4pt);
\draw[color=black] (2.2,2.3) node {$i_2$};
\draw[fill=black] (1.02,1.26) circle (1.4pt);
\draw[color=black] (.8,1.26) node {$i_3$};


\end{scriptsize}
\end{tikzpicture}
\caption{Illustration of $A$, i.e., the intersection of the geodesic unit disks centered at points in $\overline{U} = \{c, v_{u-2}, v_{u-1}\}$ as well as the disk-disk intersection points $I = \{i_1, i_2, i_3\}$.}
\label{algoPic}
\end{figure}

\begin{lemma}
Given a  simple polygon $P$, let $A$ be the non-empty intersection of a collection of geodesic unit disks in $P$ and let $\overline{\alpha\beta}$  be a line-segment in $P$, such that for all $a \in A$, $d(a, \alpha) \leq 1$ and $d(a, \beta) > 1$.   For any point $c \in \overline{\alpha\beta}$ and any disk center $q$ furthest away from $c$, $d(c, A) = 1$ if and only if either:
\begin{enumerate}
\item[a)] $d(c, q) = 2$ and $\pi(c,q) \cap \partial D(q) \in A$, or
\item[b)] $d(c, I) = 1$ and $\pi(c, q)  \cap \partial D(q) \notin A$, 
\end{enumerate}
\label{Aint}
with $I$ denoting the disk-disk intersection points on $\partial A$ and $d(c, Y) = \min_{s \in Y} d(c,s)$, for a point set $Y$ in $P$. 
\end{lemma}

\begin{proof}(Lemma~\ref{Aint})
Let $p$ be the point in $A$ closest to the point $c$ and by $q$ a center farthest from $c$. Notice that since $A$ is geodesic convex, $p$ is unique {and it lies on $\partial A$}.  We prove Lemma \ref{Aint} with the help of the following two observations.

\begin{obs}
If $p \in \partial D(q')$ for some center $q'$ and $\pi(q',c) \cap \partial D(q') \notin A$ then $p \in I$.
\label{PinI}
\end{obs}
\begin{proof}
Let $p' = \pi(q',c) \cap \partial D(q')$ and assume that $p\notin I$, thus $p$ is in the interior of all the disks defining $A$, other than $D(q')$. Furthermore,  since $D(q')$ is geodesic convex, $\pi(p, p') \cap A$ contains a point $a$, with $a \neq p$.  Since $d(p',q') = d(p, q') = 1$ and $p'$ is on the shortest path from $q'$ to $c$, by uniqueness of the shortest path, $d(c, p') < d(c, p)$.  Then, by Lemma~\ref{pollackCon}, $d(c, a) < \max\{d(c, p), d(c, p')\} = d(c, p)$, contradicting  that $p$ is the point in $A$ closest to $c$.  \qed
\end{proof}

\begin{obs}
If $p \in A\setminus I$ then $p = \pi(c, q) \cap \partial D(q)$. 
\label{PinA}
\end{obs}
\begin{proof}
Since $p \notin I$ there is a unique center $q'$ such that $p \in \partial D(q')$.    As shown in Observation \ref{PinI}, if $p \in D(q')$ and $\pi(c, q') \cap D(q') \notin A$ then $p \in I$. Therefore $\pi(c, q') \cap D(q')$ is contained in $A$ and we denote this point by $p'$.

Observe that $p'$ is contained in $\pi(c, q')$ and is at distance $1$ from $q'$, thus $d(q', c) = d(q', p') + d(p', c) = 1 + d(p', c)$.  Clearly,  $d(q',c) \leq d(q', p) + d(p, c) = 1 + d(p, c)$ and since $p$ is the closest point in $A$ to $c$, it follows that $p = p'$.  Now observe that, since $d(q', c) = 1 + d(p, c)$ and the distance between $p$ and any other center is less than $1$ (because $p \notin I$), $q'$ is the farthest center, i.e., $q' = q$.  Therefore  $p = \pi(c, q) \cap D(q)$ holds as claimed.  \qed
\end{proof}

We now prove Lemma \ref{Aint} :

"$\Rightarrow$".  We distinguish two cases based on whether $p\in A\setminus I$ or  $p\in I$.  If $p\in A \setminus I$, then $p = \pi(c, q) \cap D(q)$ as shown in Observation \ref{PinA}. Therefore, $\pi(c, q) \cap D(q) \in A$ and $d(c, q) = d(c, p) + d(p, q) = 2$, and thus condition a) holds.  If $p \in I$ and $\pi(c, q) \cap D(q) \notin A$ condition b) holds.  Otherwise if $p \in I$ and   $\pi(c, q) \cap D(q) \in A$, let $p' = \pi(c, q) \cap D(q)$. We can write the distance $d(c, q)$ as $d(c, p') + d(p', q)$. Since $p'$ lies on $\partial D(q)$, $d(p',q) = 1$ and since $d(c,A) = 1$ this implies that $d(c,p') \geq 1$. Therefore, distance $d(c,q) \geq 2$. 

By the triangle inequality it also holds that $d(c, q) \leq d(c, p) + d(p, q)$. Since $p$ is the closest point in $A$ to $c$, $d(c,p) = 1$ by hypothesis. Since $p$ lies on $\partial D(q)$, $d(p,q) = 1$. Therefore,  $d(c, q) \leq 2$ and combining this with $d(c,q) \geq 2$ from above, $d(c, q) = 2$ and again condition a) holds.

\noindent "$\Leftarrow$"
\noindent a) Let $p' = \pi(c,q) \cap \partial D(q)$.  Then $d(c, A) \leq d(c, p') = d(c, q) - d(q, p') = 1$.  If $d(c,A) < 1$, by definition $d(c,p) < 1$. Since $p \in A$, $d(p,q) \leq 1$ and by the triangle inequality, $d(c, q) \leq d(c, p) + d(p, q) < 2$ which contradicts $d(c, q) = 2$.

\noindent b) Since $d(c, I) \leq 1$ and $I \subseteq A$, obviously $d(c, A) \leq 1$.  For $p \in A$ the closest point to $c$ in $A$, assume that $d(c, p)< 1$.  Since $d(c, I) = 1$, $p \notin I$.  Therefore, by Observation \ref{PinA}, $p = \pi(c,q') \cap \partial D(q')$, and thus this intersection is in $A$ contradicting the hypothesis.
\qed
\end{proof}

We use the following steps to determine $c'$ on $e = \overline{v_{u-1}v_{u}}$.\\
Step 1)  Find the point $x_1$ on $e$ closest to $v_u$, whose distance to its furthest point $q$ in $\overline{U}$ is exactly $2$ and $\pi(x_1,q) \cap \partial D(q) \in A$, if such a point $x_1$ exists. \\
Step 2) Find the point $x_2$ on $e$ closest to $v_u$, whose distance to its closest point in $I$ is exactly $1$.  \\
Step 3) Set $c' \gets x_2$ if $x_1$ does not exist or $d(x_2,v_u) < d(x_1,v_u)$. In this case we add the point in $I$ closest to $c'$ as the new disk center to the set $S$ of centers. Otherwise $c' \gets x_1$ and the point $\pi(x_1,q) \cap \partial D(q)$ is the new disk center which gets added to $S$.\\

Note that since $v_{u-1}$ will be covered in this iteration and $v_{u}$ won't be covered, $d(v_{u-1}, A) \leq 1 < d(v_{u}, A)$. By continuity of the geodesic distance, there is a point $c'$ on $e$, with $d(c', A) = 1$ and thus by Lemma~\ref{Aint} either $x_1$ or $x_2$ exists.\\

In Step 1, to find $x_1$ if it exists, we construct the (refined) furthest-site geodesic Voronoi diagram of the sites $\overline{U}$ in $Q$ and traverse the Voronoi vertices $\gamma_1, \ldots, \gamma_m$ on $e$, ordered in the direction from $v_{u}$  to $v_{u-1}$ and set $\gamma_{m+1} =  v_{u-1}$. For each such vertex we check in $O(\log |Q|)$ time whether the distance to (one of) its furthest site(s) is at most $2$, using an $O(\log |Q|)$ time shortest path query (\cite{Guibas1989126}) after pre-processing $Q$ in $O(|Q|)$ time. Once we find the first $\gamma_j$ with distance at most $2$, if it exists, this determines a sub-segment $\overline{\gamma_j\gamma_{j-1}}$ on $e$ containing a point $x$ at distance exactly $2$ from its furthest site $q$. Note that since the shortest paths to the furthest site $q$ have the same combinatorial structure for all points on the refined Voronoi edge $\overline{\gamma_j\gamma_{j-1}}$, we find the point at distance $2$ to $q$ in constant time since we stored the anchor point $a_{\overline{\gamma_j\gamma_{j-1}}}(q)$ at the edge $\overline{\gamma_j\gamma_{j-1}}$ (see Section 2.1). We check if $\pi(x,q) \cap \partial D(q) \in A$, by computing $D(q)$ in time $O(|Q|)$ using \cite{lineshortpath} and finding in $O(\log|Q|)$ time the arc $\alpha$ of $D(q)$ separating $q$ from $x$. We traverse the edges of $\pi(x,q)$ and for each edge we test in $O(1)$ time if it intersects $\alpha$. Denoting the intersection point by $p$, we check if $p \in A$, by computing the shortest path tree to the sites in $U$ and test if the distance to all sites is at most $1$ in time $O(|Q|)$. If this intersection is in $A$, we set $x_1$ to $x$.\\

\vspace{-9pt}
\begin{claim}There can be at most two points on $e$ that have distance exactly $2$ from their respective furthest site; if there are two such points, one of them must be $v_{u-1}$.  
\label{claim}
\end{claim}
\vspace{-3pt}

We prove this claim using the following lemma.
\begin{lemma}[Lemma 1 \cite{pollackGeod}; see also Lemma 2.2.1 \cite{BorisFurthest}]
Given three points $a,b,c$ in a simple polygon, for $x \in \pi(b,c)$, the distance $d(a,x)$ is a convex function on $\pi(b,c)$, with  $d(a,x) < \max\{d(a,b), d(a,c)\}$.
\label{pollackCon}
\end{lemma}

\begin{proof} [of Claim]
Assume that there are two points $p_1$ and $p_2$ on $e\setminus\{v_{u-1}\}$ that are at distance $2$ from their respective furthest sites, with $p_1$ closer to $v_{u-1}$ than $p_2$, thus $p_1 \in \overline{v_{u-1}p_2} \setminus \{v_{u-1}, p_2\}$.  Let $q_1$ be a center furthest away from $p_1$. Clearly $d(q_1, v_{u-1}) \leq 2$ since both $q_1$ and $v_{u-1}$ are at distance at most $1$ from any point in $A$. Since $d(q_1, p_2) \leq 2$ and $p_1 \in \overline{v_{u-1}p_2} \setminus \{v_{u-1}, p_2\}$, by Lemma~\ref{pollackCon} $d(q_1, p_1) < 2$, contradicting the assumption that $d(q_1, p_1) = 2$. \qed
\end{proof}

According to the above claim, the only other candidate for $x_1$ is $v_{u-1}$. 
Thus, if $\pi(x_1,q) \cap \partial D(q) \notin A$ we check in $O(\log |Q|)$ time if the point $v_{u-1}$ is at distance exactly $2$ from its furthest site and if so, we set $x_1$ to $v_{u-1}$. If $x_1$ exists $\partial P[c,x_1]$ can be covered with one geodesic unit disk, because the point $\pi(x_1,q) \cap \partial D(q)$ has distance exactly $1$ to $x_1$ and lies in $A$.   \\

In Step 2, to find $x_2$, we first construct the set $I$ of the disk-disk intersection points of $A$; we do this without explicitly computing $A$. To construct $I$, we look at the furthest-site geodesic Voronoi diagram of the sites $\overline{U}$ in $Q$ constructed in the Step 1. Since any point in $I$ has two points in $\overline{U}$ at distance $1$, every point in $I$ lies on a Voronoi edge. For every site $s \in \overline{U}$ we look at the refined edges of $\sigma(s)$ and for such edge $e$ we access its anchor point $a_e(s)$ as well as the distance from $s$ to the endpoints of $e$, in constant time. We test if there is a point on $e$ having distance $1$ to $s$, again in $O(1)$ time. If such a point exists then this is a disk-disk intersection point and we add it to $I$. Since we need constant time for each refined Voronoi edge, $I$ can be computed in total time $O(|Q|)$.

Having computed $I$, we construct the closest-site geodesic Voronoi diagram of the sites  $I$  in $Q$. We traverse the Voronoi vertices $\gamma_1, \ldots, \gamma_m$ on $e$, ordered in the direction from $v_{u}$  to $v_{u-1}$ and set $\gamma_{m+1} =  v_{u-1}$. For each such vertex we check whether the distance to (one of) its closest site(s) is at most $1$ again by an $O(\log n)$ time shortest path distance query. Once we find the first such vertex $\gamma_j$ on $e = \overline{v_{u-1}v_{u}}$, if it exists, we have determined a sub-segment $\overline{\gamma_j\gamma_{j-1}}$ on $e$ where $x_2$ lies. Letting $i \in I$ be the corresponding closest site, by Lemma  \ref{pollackCon}, we find the point in $\overline{\gamma_j\gamma_{j-1}}$  at distance $1$ from $i$ by computing the intersection point of a geodesic unit disk centered at $i$ with $\overline{\gamma_j\gamma_{j-1}}$, in time $O(|Q|)$, using the funnel algorithm of \cite{lineshortpath}.

There can be at most two points on $e$ that have distance exactly $1$ from $i$; if there are two such points, one of them must be $v_{u-1}$.  This can be seen directly from the fact that $d(i, v_{u-1}) \leq 1$, and Lemma~\ref{pollackCon}.  We set $x_2$ to the one closer to $v_u$. It is easy to see that $x_2$ is feasible, i.e., $\partial P[c,x_2]$ can be covered with one geodesic unit disk, because $d(i, \overline{U}) \leq 1$ and $d(i,x_2)=1$.\\

In Step 3, $c' \gets x_2$, if either $x_1$ does not exist or $d(x_2,v_u) < d(x_1,v_u)$, thus  $c'$ indeed extends $\Gamma$ maximally because $x_2$ is the point on $e$ closest to $v_u$ having distance exactly $1$ to the closest point in $I$, i.e., to the center of the geodesic unit disk placed in this iteration. Otherwise $c' \gets x_1$ and $x_1$ is the point on $e$ closest to $v_u$ having distance exactly $2$ to the furthest center in $\overline{U}$; any point on $e$ closer to $v_u$ has distance larger than $2$ from that center and is thus infeasible.\\

\noindent {\em Computational Complexity / Summary.} 
Constructing $Q$ takes time $O(|Q| \log n)$ as argued in the \textsc{TestCover}$(c, v)$ paragraph before.
Step 1 needs $O(|Q| \log |Q|)$ time to construct the geodesic furthest-site Voronoi diagram of $\overline{U}$ in $Q$ and $O(|Q| \log |Q|)$  time to find a sub-segment of the edge $e$ possibly containing $x_1$, since there are only $O(|Q|)$ Voronoi vertices in total and we spend $O(\log |Q|)$ on them for finding the sub-segment. The last step is to test if $\pi(x,q) \cap \partial D(q) \in A$, which takes time $O(|Q|)$ as argued above.

In Step 2, we spend $O(|Q|)$ time to construct the set $I$ and $O(|Q| \log |Q|)$ time to construct the geodesic closest-site Voronoi diagram of the sites $I$. We then traverse edge $e$ in $O(|Q| \log |Q|)$ time to find a sub-segment of the edge $e$ possibly containing $x_2$, and determine $x_2$ on this sub-segment in $O(|Q|)$ time.

Thus the overall time spent in \textsc{AugmentShort}  is $O(|Q|\log n)$.

\subsubsection{Total Running Time.}

Let $\mathcal{Q}$ be the set of all polygons constructed throughout the whole execution of \textsc{ContiguousGreedy}. In each polygon $Q\in \mathcal{Q}$ we spend $O(|Q|\log n)$ time in \textsc{TestCover} and possibly $O(|Q| \log n)$ time in \textsc{AugmentShort} as argued above. Since in each iteration of \textsc{ContiguousGreedy}, $\Gamma$ is extended to cover at least one new polygon vertex, there are at most $n$ iterations of the main $while$ loop. Furthermore,  covering long segments of $\partial P$ takes total time $O(k)$. Since according to Lemma \ref{log_poly}, $\sum_{Q \in \mathcal{Q}}|Q| = O(n \log n)$, the running time of \textsc{ContiguousGreedy} is $O(n \log^2 n + k)$. 
\vspace{5pt}
\begin{lemma}
$\sum_{Q \in \mathcal{Q}} |Q| = O(n \log n)$.
\label{log_poly}
\end{lemma}

\begin{proof}
Each polygon of $\mathcal{Q}$ constructed in the \textsc{ContiguousCover} algorithm has the form $Q = \partial P[c, w]  \circ \pi(w,c)$, with $c$ an arbitrary point on $\partial P$ and $w$ a vertex of $P$. We call $\partial P[c, w]$ the $\partial$-portion and $\pi(c, w)$ the $\pi$-portion of the polygon $Q$.  Notice that  every polygon constructed in \textsc{AugmentShort} was also constructed in a \textsc{TestCover} call and thus it suffices to bound the number of polygons constructed in all \textsc{TestCover} calls.

Observe that $|\mathcal{Q}| = O(n \log n)$, since in each iteration, $\Gamma$ is extended to cover at least one new vertex, thus there are at most $n$ iteration, and in each iteration we construct $O(\log n)$ polygons during Exponential and Binary Search.
Observe that if every vertex of $P$ is contained in $O(\log n)$ polygons of $\mathcal{Q}$ then $\sum_{Q \in \mathcal{Q}} |Q| = O(n \log n)$. This holds because for each $Q \in \mathcal{Q}$ there is at most one vertex of $Q$ which is not a vertex in $P$, namely the point $c$. \\ Since $v_1$ is covered both in the first and last iteration of the algorithm, we are pessimistically bounding the number of polygons containing $v_1$ by $|\mathcal{Q}| = O(n \log n)$. To then prove the lemma it is enough to show that every vertex of $P$ except $v_1$ is contained in $O(\log n)$ polygons of $\mathcal{Q}$. For that we fix a vertex $v_k$, with $1<k \leq n$, and show that $v_k$ appears in the $\partial$-portion of $O(\log n)$ polygons and $v_k$ appears in the $\pi$-portion of $O(\log n)$ polygons of $\mathcal{Q}$.  

To bound the number of appearances of $v_k$ on the $\partial$-portion of a polygon we fix the unique iteration $i^*$ in which  $v_k$ is first covered.  Since \textsc{TestCover} is used as a predicate in Exponential and Binary search, in iteration $i^*$ it is called $O(\log n)$ times and thus $v_k$ appears in $O(\log n)$ polygons during this iteration. Observe, that in subsequent iterations, when $i>i^*$, vertex $v_k$ is not part of the $\partial$-portion of any constructed polygon. For an iteration $i < i^*$, let $v_{u_i}$ be the first uncoverable vertex (denoted by $v_u$ in the algorithm) found in iteration $i$, thus  $u_i \leq k$; let $q_i$ be the number of polygons in which $v_k$ appears on the $\partial$-portion during this iteration $i$. Also observe that $u_{i-1}$ is the index of the first  vertex of $P$ covered in iteration $i$.
We claim that 
\begin{align}
k - u_i \leq \frac{k-u_{i-1}}{2^{q_i-1}},  \text{for any} \; 1\leq i < i^*
\label{loginterval}
\end{align}
and defining $u_0 = 1$, implies that $\sum^ {i^*-1}_{i = 1} q_i  \leq \log k \leq \log n$.

For $q_i = 0$, inequality (\ref{loginterval}) holds trivially. Otherwise, since $v_k$ is not covered during this iteration, Exponential Search stops after the first time $v_k$ appears on the $\partial$-portion of a constructed polygon. This leaves a search interval of size at most $k-u_{i-1}$. During Binary Search, there are exactly $q_i-1$ search intervals which contain both $v_u$ and $v_k$. Since the interval size is halved at each step and all search intervals containing both $v_{u_i}$ and $v_k$ have size at least $k-u_i$,  inequality (\ref{loginterval}) follows.

So far we have shown that $v_k$ appears on the $\partial$-portion of $O(\log k)$ polygons in $\mathcal{Q}$ before iteration $i^*$, $O(\log n)$ times during iteration $i^*$ and does not appear in subsequent iterations.  Therefore, all together, $v_k$ appears on  $\partial$-portions of $O(\log n)$ polygons in $\mathcal{Q}$.

To bound the number of appearances of $v_k$ on the $\pi$-portion of a polygon, let $\mathcal{Q}_k \subseteq \mathcal{Q}$ be the set of polygons containing $v_k$ on their $\pi$-portion but not on the $\partial$-portion.  By Observation~\ref{obsPath} below, any two polygons in $\mathcal{Q}_k$ intersect on their $\partial$-portion because they both contain $v_k$ on their $\pi$-portion. Since by construction the $\partial$-portion of each polygon $Q$ ends with a vertex, any two polygons in $\mathcal{Q}_k$ have a vertex in common on their $\partial$-portion.  This is true because the $\partial$-portion of those polygons are subsequences of $(v_1, ..., v_{n}, v_1)$ and it is easy to see that there is a vertex $v_{k'}$ that belongs to the $\partial$-portion of  all $Q \in \mathcal{Q}_k$.  Since $v_{k'}$ appears on  $\partial$-portions of $O(\log n)$ polygons, $|\mathcal{Q}_k| = O(\log n)$.
\qed
\end{proof}

\begin{obs}
For $a,b,c,d$ four distinct points on $\partial P$, if $\pi(a,b) \cap \pi(c,d)$ contains a polygon vertex not contained in $\partial P[a,b] \cup \partial P[c,d]$, then $\partial P[a,b] \cap \partial P[c,d] \neq \emptyset$. 
\label{obsPath}
\end{obs}

\begin{figure}

\definecolor{cqcqcq}{rgb}{0.75,0.75,0.75}
\center
\begin{tikzpicture}[line cap=round,line join=round,>=triangle 45,x=1.5cm,y=1.5cm]
\clip(0,0) rectangle (8,2.4);
\draw (0.4,0.46)-- (1.44,1.84)-- (1.96,1.13)-- (2.64,1.89)-- (3.29,1.47)-- (3.86,1.74)-- (3.36,1.94)-- (4.92,2.05)-- (5.47,1.41)-- (6.65,0.95);
\draw (0.33,0.25)-- (1.41,0.86)-- (2.41,0.42)-- (2.53,0.7)-- (3.57,1.2)-- (3.26,0.15)-- (4.03,0.25)-- (4.17,0.64)-- (3.83,0.53)-- (3.68,0.73)-- (4.65,1.21)-- (4.97,0.7)-- (6.08,0.67)-- (7.03,0.73);
\draw [dash pattern=on 2pt off 2pt] (0.4,0.46)-- (1.41,0.86)-- (3.57,1.2)-- (4.65,1.21)-- (6.65,0.95);
\fill[color=gray,fill=gray,fill opacity=0.1] (3.57,1.2) -- (4.65,1.21) -- (3.68,0.73) -- (3.83,0.53) -- (4.17,0.64) -- (4.03,0.25) -- (3.26,0.15) -- cycle;

\fill[pattern color=cqcqcq,fill=cqcqcq,pattern=north east lines] (4.65,1.21) -- (4.96,0.67) -- (6.08,0.67) -- (7.05,0.73) -- (6.67,0.97) -- cycle;


\begin{scriptsize}

\draw[color=black] (4.21,2.14) node {$\partial P[a,b]$};
\draw[color=black] (2.44,1.12) node {$\pi(a,b)$};

\fill [color=black] (0.4,0.46) circle (1.5pt);
\draw[color=black] (0.41,0.63) node {$a$};
\fill [color=black] (1.44,1.84) circle (1.5pt);
\fill [color=black] (1.96,1.13) circle (1.5pt);
\fill [color=black] (2.64,1.89) circle (1.5pt);
\fill [color=black] (3.29,1.47) circle (1.5pt);
\fill [color=black] (3.86,1.74) circle (1.5pt);
\fill [color=black] (3.36,1.94) circle (1.5pt);
\fill [color=black] (4.92,2.05) circle (1.5pt);
\fill [color=black] (5.47,1.41) circle (1.5pt);
\fill [color=black] (6.65,0.95) circle (1.5pt);
\draw[color=black] (6.73,1.13) node {$b$};
\fill [color=black] (0.33,0.25) circle (1.5pt);
\fill [color=black] (1.41,0.86) circle (1.5pt);
\fill [color=black] (2.41,0.42) circle (1.5pt);
\fill [color=black] (2.53,0.7) circle (1.5pt);
\fill [color=black] (3.57,1.2) circle (1.5pt);
\draw[color=black] (3.67,1.27) node {$v$};
\fill [color=black] (3.26,0.15) circle (1.5pt);
\fill [color=black] (4.03,0.25) circle (1.5pt);
\fill [color=black] (4.17,0.64) circle (1.5pt);
\draw[color=black] (4.28,0.71) node {$d$};
\fill [color=black] (3.83,0.53) circle (1.5pt);
\fill [color=black] (3.68,0.73) circle (1.5pt);
\fill [color=black] (4.65,1.21) circle (1.5pt);
\draw[color=black] (4.76,1.29) node {$w$};
\fill [color=black] (4.97,0.7) circle (1.5pt);
\fill [color=black] (6.08,0.67) circle (1.5pt);
\draw[color=black] (6.2,0.74) node {$c$};
\fill [color=black] (7.03,0.73) circle (1.5pt);
\draw[color=black] (5.36,0.93) node {$R$};
\draw[color=black] (5.36,0.53) node {$\partial P[c,d]$};
\draw[color=black] (3.86,0.99) node {$Q$};
\end{scriptsize}
\end{tikzpicture}
\caption{Illustration of the proof of Observation~\ref{obsPath}.}
\label{pathPic}
\end{figure}
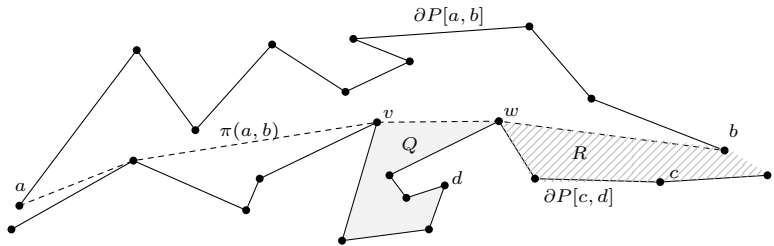

\begin{proof}

Let $v$ be a vertex contained in $\pi(a,b) \cap \pi(c,d)$  not in $\partial P[a,b] \cup \partial P[c,d]$ and assume for contradiction that $\partial P[a,b] $ and $ \partial P[c,d]$ are disjoint. 
Then either $\partial P[c,d] \subset \partial P[a,v]$ or $\partial P[c,d] \subset \partial P[v,b]$. W.l.o.g.  assume  $\partial P[c,d] \subset \partial P[v,b]$. For $w$ the successor vertex of $v$ in $\pi(a,b)$, let $Q$ be the simple polygon bounded by $\overline{vw} \circ \partial P[w, v]$. If both $c,d$ are contained in $Q$, meaning $w \notin \partial P[c,d]$, since $v$ is a convex vertex in $Q$, it holds that $ v \notin \pi(c,d)$, a contradiction. Otherwise, let $R$ be the geodesic convex set bounded by $\pi(w, b) \circ \partial P[b, w]$. If both $c,d$ are contained in $R$, then by geodesic convexity, $\pi(c,d) \subseteq R$ and thus $v \notin \pi(c,d)$. Otherwise $d \in Q$ and $c \in R$ as shown in Fig.~\ref{pathPic}. Since in that case $Q \cap  R = \{w\}$ and $R \cup Q$ is again a geodesic convex set, $\pi(c,d) =\pi(c,w) \circ \pi(w,d)$. Again, since $v$ is a convex vertex in $Q$, $v \notin \pi(w,d)$, and thus $v$ not in $\pi(c,d)$, a contradiction.
\qed

\end{proof}

\section{Approximation Ratio}
\label{approx}

Let $OPT$ denote a set of geodesic unit disks optimally covering $\partial P$. In order to prove the $2$-approximation we prove the existence of a {\em coloring} for  $\partial P$ using $|OPT|$ distinct colors and introducing at most $\max\{2|OPT| - 2,1\}$ monochromatic boundary portions. We then show that  \textsc{ContiguousGreedy} uses at most one disk per monochromatic boundary portion (plus possibly one additional disk for the unique monochromatic boundary portion containing $v_1$), which implies the $2$-approximation factor of  \textsc{ContiguousGreedy}. 

{A {\em coloring} of $\partial P$ is a function $\gamma: \partial P \rightarrow \mathbb{N}$.  The number of colors used by $\gamma$ is defined as the cardinality of the image of $\gamma$}. A {\em block} is a connected component of $\partial P$ colored with a single color.  We let $\partial P_i$ denote the subset of the polygon boundary colored with color $i$ and we call each connected component of $ \partial P \setminus \partial P_i$ a \emph{pocket} of $\partial P$ induced by {color} $i$ (see Fig.~\ref{colorExFig}(b)).

A coloring of $\partial P$  is called {\em crossing-free} if for any two distinct colors {$i, j$, it holds that $\partial P_j$} is contained in a single pocket induced by color $i$.

For a collection $\mathcal{D} = \{D_1, \ldots, D_k\}$ of disks covering $\partial P$, a \emph{disk-coloring} of $\partial P$ w.r.t. $\mathcal{D}$ is a function $\gamma_{\mathcal{D}}: \partial P \rightarrow \{1, \ldots, k\}$, such that $\gamma(x) = i \Rightarrow x \in  D_i$, i.e., a point on $\partial P$  can only be colored with one of the indices of the disks covering it (see Fig.~\ref{colorExFig}(a)).

\definecolor{ffqqqq}{rgb}{1.0,0.0,0.0}
\definecolor{qqwuqq}{rgb}{0.0,0.39215686274509803,0.0}
\definecolor{qqqqff}{rgb}{0.0,0.0,1.0}
\definecolor{cqcqcq}{rgb}{0.7529411764705882,0.7529411764705882,0.7529411764705882}
\definecolor{qqwuqq}{rgb}{0.0,0.39215686274509803,0.0}
\definecolor{ffqqqq}{rgb}{1.0,0.0,0.0}
\definecolor{qqqqff}{rgb}{0.0,0.0,1.0}
\definecolor{orange}{rgb}{1.0,0.5,0.0}
\definecolor{brown}{rgb}{0.5,.5,0.0}
\definecolor{yqyqyq}{rgb}{0.5019607843137255,0.5019607843137255,0.5019607843137255}
\definecolor{cqcqcq}{rgb}{0.7529411764705882,0.7529411764705882,0.7529411764705882}

\vspace{-15pt}

\begin{figure}[]
\center
\begin{tabular}{cc}

 \subfloat[] {

\begin{tikzpicture}[line cap=round,line join=round,>=triangle 45,x=1.0cm,y=1.0cm,scale=1.3]
\clip(-0.15,0.2) rectangle (3.0,2);
\draw [dotted,color=qqqqff,fill=qqqqff,fill opacity=0.1] (1.6407502437867678,0.7262099843520922) circle (0.5cm);
\draw [color=qqwuqq,fill=qqwuqq,fill opacity=0.1] (2.3372699835759234,0.7288531726287284) circle (0.5cm);
\draw [dash pattern=on 1pt off 1pt on 1pt off 4pt,color=orange,fill=orange,fill opacity=0.1] (1.6481657819908717,1.4707564648272209) circle (0.5cm);
\draw [color=brown,fill=brown,fill opacity=0.1] (1.1442305039976122,0.7262099843520922) circle (0.5cm);
\draw [dash pattern=on 1pt off 1pt on 1pt off 4pt,color=ffqqqq,fill=ffqqqq,fill opacity=0.1] (0.42379009724472,0.7943597525584467) circle (0.5cm);
\draw [line width=1.2000000000000002pt,dash pattern=on 1pt off 1pt on 1pt off 4pt,color=ffqqqq] (0.5,0.5)-- (0.6967952937271886,0.5032210625048962);
\draw [line width=1.2000000000000002pt,dash pattern=on 1pt off 1pt on 1pt off 4pt,color=ffqqqq] (0.5,0.5)-- (0.0,1.0);
\draw [line width=1.2000000000000002pt,dash pattern=on 1pt off 1pt on 1pt off 4pt,color=ffqqqq] (0.0,1.0)-- (0.7246119238309283,1.00392040437221);
\draw [line width=1.2000000000000002pt,color=brown] (0.6967952937271886,0.5032210625048962)-- (1.1882224255599232,0.5032210625048962);
\draw [line width=1.2000000000000002pt,dotted,color=qqqqff] (1.1882224255599232,0.5032210625048962)-- (1.8882742831707053,0.5032210625048962);
\draw [line width=1.2000000000000002pt,dash pattern=on 1pt off 1pt on 1pt off 4pt,color=orange] (1.4942053567010596,0.9992842993549201)-- (1.499540803585664,1.2735138372423407);
\draw [line width=1.2000000000000002pt,dash pattern=on 1pt off 1pt on 1pt off 4pt,color=orange] (1.499540803585664,1.2735138372423407)-- (1.994904698568374,1.7735138372423407);
\draw [line width=1.2000000000000002pt,dash pattern=on 1pt off 1pt on 1pt off 4pt,color=orange] (1.994904698568374,1.7735138372423407)-- (1.999540803585664,1.1105508197698788);
\draw [line width=1.2000000000000002pt,color=qqwuqq] (1.999540803585664,1.1105508197698788)-- (2.0,1.0);
\draw [line width=1.2000000000000002pt,color=qqwuqq] (2.0,1.0)-- (2.5,1.0);
\draw [line width=1.2000000000000002pt,color=qqwuqq] (2.5,1.0)-- (2.5,0.5);
\draw [line width=1.2000000000000002pt,color=qqwuqq] (2.5,0.5)-- (1.8882742831707053,0.5032210625048962);
\draw [line width=1.2000000000000002pt,color=brown] (1.4942053567010596,0.9992842993549201)-- (0.7246119238309283,1.00392040437221);
\begin{scriptsize}
\draw[color=qqqqff] (1.7,0.4) node {$D_4$};
\draw[color=qqwuqq] (2.652873888891698,0.7410756695859958) node {$D_5$};
\draw[color=orange] (1.4436826803576182,1.8148374627642583) node {$D_3$};
\draw[color=brown] (1,0.4) node {$D_2$};
\draw[color=ffqqqq] (0.16677676414562995,0.68303449157636) node {$D_1$};


\end{scriptsize}

\end{tikzpicture}

} &

\subfloat[]{

\begin{tikzpicture}[line cap=round,line join=round,>=triangle 45,x=1.0cm,y=1.0cm,scale=1.3]
\clip(-0.15,-1.55) rectangle (3.0, 0);

\draw [line width=1.2000000000000002pt,dash pattern=on 1pt off 1pt on 1pt off 4pt,color=ffqqqq] (0.5,0.5 - 1.75 )-- (0.6967952937271886,0.5032210625048962   - 1.75 );
\draw [line width=1.2000000000000002pt,dash pattern=on 1pt off 1pt on 1pt off 4pt,color=ffqqqq] (0.5,0.5 - 1.75 )-- (0.0,1.0 - 1.75 );
\draw [line width=1.2000000000000002pt,dash pattern=on 1pt off 1pt on 1pt off 4pt,color=ffqqqq] (0.0,1.0 - 1.75 )-- (0.7246119238309283,1.00392040437221 - 1.75 );

\draw [line width=1.2000000000000002pt,dotted,color=qqqqff] (1.1882224255599232,0.5032210625048962 - 1.75 )-- (1.8882742831707053,0.5032210625048962 - 1.75 );
\draw [line width=1.2000000000000002pt,dash pattern=on 1pt off 1pt on 1pt off 4pt,color=orange] (1.4942053567010596,0.9992842993549201 - 1.75 )-- (1.499540803585664,1.2735138372423407 - 1.75 );
\draw [line width=1.2000000000000002pt,dash pattern=on 1pt off 1pt on 1pt off 4pt,color=orange] (1.499540803585664,1.2735138372423407 - 1.75 )-- (1.994904698568374,1.7735138372423407 - 1.75 );
\draw [line width=1.2000000000000002pt,dash pattern=on 1pt off 1pt on 1pt off 4pt,color=orange] (1.994904698568374,1.7735138372423407 - 1.75 )-- (1.999540803585664,1.1105508197698788 - 1.75 );
\draw [line width=1.2000000000000002pt,color=qqwuqq] (1.999540803585664,1.1105508197698788 - 1.75 )-- (2.0,1.0 - 1.75 );
\draw [line width=1.2000000000000002pt,color=qqwuqq] (2.0,1.0 - 1.75 )-- (2.5,1.0 - 1.75 );
\draw [line width=1.2000000000000002pt,color=qqwuqq] (2.5,1.0 - 1.75 )-- (2.5,0.5 - 1.75 );
\draw [line width=1.2000000000000002pt,color=qqwuqq] (2.5,0.5 - 1.75 )-- (1.8882742831707053,0.5032210625048962 - 1.75 );

\end{tikzpicture}

}
\end{tabular}
\caption{(a) A crossing disk-free coloring of $\partial P$. (b)  The two pockets induced by color $2$. }
\label{colorExFig}
\end{figure}
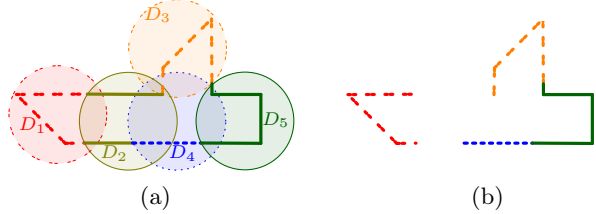

\vspace{-25pt}

\begin{definition}
For a coloring $\gamma$, two of its colors $r$ and $b$ \emph{cross} each other, if there are two pockets induced by color $r$ containing blocks of color $b$.
\end{definition}

Observe that if two colors $r$ and $b$ cross each other, there are at least two blocks $B^1_r, B^2_r$ of color $r$ and two blocks $B^1_b, B^2_b$ of color $b$ such that sequence of blocks $B^1_r, B^1_b, B^2_r, B^2_b$ occurs in clockwise order on $\partial P$ as shown in Fig. \ref{2pockets}.

\begin{lemma}
In any disk-coloring, if two colors $r$ and $b$ cross each other, one of the following holds: 1) There exists a pocket induced by color $r$ which contains blocks of color $b$ and all these blocks can be re-colored with color $r$, s.t. the resulting coloring is still a disk-coloring. 2) There exists a pocket induced by color $b$ which contains blocks of color $r$ and all these blocks can  be re-colored with color $b$, s.t. the resulting coloring is still a disk-coloring.
\label{crossLem}
\end{lemma}
\begin{proof}

Suppose this is not possible. Since neither $B^1_r$ nor $B^2_r$ can be colored with $b$, there are  points $\alpha_r \in B^1_r$ and $\beta_r \in B^2_r$ which lie outside of disk $D_b$. If we denote the center of $D_b$ by $c_b$, it therefore holds that $d(c_b, \alpha_r) > 1$ and $d(c_b, \beta_r) > 1$ (see Fig. \ref{2pockets}). Analogously, there are two points $\alpha_b \in B^1_b$ and $\beta_b \in B^2_b$, s.t. $\alpha_b$ and $\beta_b$ can not be colored with color $r$. This again implies that both points lie outside of disk $D_r$ centered at $c_r$ and thus  $d(c_r, \alpha_b) > 1$ and $d(c_r, \beta_b) > 1$.

\begin{figure}
\center
\definecolor{ffqqqq}{rgb}{1,0,0}
\definecolor{cqcqcq}{rgb}{0.75,0.75,0.75}
\begin{tikzpicture}[line cap=round,line join=round,>=triangle 45,x=1.0cm,y=1.0cm, scale = 1]
\clip(-1,0.2) rectangle (3,3.9);
\draw [line width=1.2pt,color=red] (0.22,3.22)-- (0.78,3.7)-- (1.6,3.2);
\draw [line width=1.2pt,color=blue] (-0.2,2.68)-- (-0.7,1.92)-- (-0.24,0.56) ;
\draw [line width=1.2pt,color=red] (0.4,0.56)-- (1.82,0.56);
\draw [line width=1.2pt,color=blue]  (2.58,2.96)-- (2.38,2.22)-- (2.66,0.88);
\draw [dotted,color=black] (1.12,2.3)-- (0.64,3.58);
\draw [dotted,color=black] (1.12,2.3)-- (1.64,0.56);
\draw [dotted,color=black] (0.86,1.9)-- (2.64,0.96);
\draw [dotted,color=black] (0.86,1.9)-- (-0.63,1.66);
\begin{scriptsize}
\draw [fill=black] (1.12,2.3) circle (1.4pt);
\draw[color=black] (1.44,2.4) node {$c_r$};
\draw [fill=black] (0.86,1.9) circle (1.4pt);
\draw[color=black] (0.64,2.14) node {$c_b$};
\draw [fill=black] (0.64,3.58) circle (1.4pt);
\draw[color=black] (0.45,3.7) node {$\beta_r$};
\draw [fill=black] (1.64,0.56) circle (1.4pt);
\draw[color=black] (1.8,0.3) node {$\alpha_r$};
\draw [fill=black] (2.64,0.96) circle (1.4pt);
\draw[color=black] (2.74,0.66) node {$\beta_b$};
\draw [fill=black] (-0.63,1.66) circle (1.4pt);
\draw[color=black] (-0.75,1.3) node {$\alpha_b$};
\draw [fill=black] (1.32,1.66) circle (1.4pt);
\draw[color=black] (1.22,1.44) node {p};

\draw[color=black] (0.8,0.35) node {$B^1_r$};
\draw[color=black] (2.8,2.) node {$B^2_b$};
\draw[color=black] (1.5,3.5) node {$B^2_r$};
\draw[color=black] (-0.8,2.3) node {$B^1_b$};

\end{scriptsize}
\end{tikzpicture}
\caption{Illustration of the four alternating blocks $B^1_r, B^1_b, B^2_r, B^2_b$  and the corresponding points $\alpha_r, \beta_r$ and $\alpha_b$, $\beta_b$; the disk centers $c_r$ and $c_b$, as well as the intersection point $p$ of  $\pi(c_r, \alpha_r)$ and $\pi(c_b, \beta_b)$.}
\label{2pockets}
\end{figure}
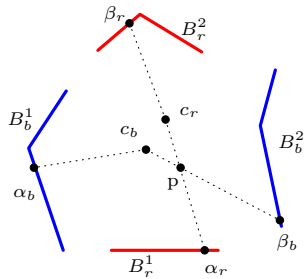

\begin{lemma}
For any collection of disks covering $\partial P$, there exists a crossing free disk-coloring of $\partial P$.
\label{crossFreeExistsLem}
\end{lemma}
\begin{proof}
Consider the four paths $\pi(c_r, \alpha_r)$, $\pi(c_r, \beta_r)$, $\pi(c_b,\alpha_b)$ and $\pi(c_b, \beta_b)$. Due to the alternating arrangement of the four blocks $B^1_r, B^1_b, B^2_r, B^2_b$  -- and therefore of $\alpha_r, \alpha_b, \beta_r, \beta_b$, on the polygon boundary, one of the paths from $c_r$ must intersect with one of the paths from $c_b$.  Assume w.l.o.g. that $\pi(c_r, \alpha_r)$ intersects $\pi(c_b, \beta_b)$ and let $p$ be an intersection point. Again, w.l.o.g., assume that $d(c_r, p) \leq d(c_b, p)$.  Then, by the triangle inequality $d(c_r, \beta_b) \leq d(c_b, \beta_b) \leq 1$ contradicting our assumption that $d(c_r, \beta_b) > 1$.
\qed
\end{proof}

We are now going to prove Lemma~\ref{crossFreeExistsLem}, which states that for any collection of disks covering $\partial P$, there exists a crossing free disk-coloring of $\partial P$.\\

For a given disk-coloring w.r.t. a collection of disks $\mathcal{D}$, we let $l_{ij}$ be the number of pockets induced by color $i$ which contain blocks of color $j$; observe that $l_{ij} = l_{ji}$. We refer to $$\sum_{1 \leq i < j \leq |\mathcal{D}|} (l_{ij} -1) $$ as the \emph{crossing number} of the disk-coloring. By definition it holds that the crossing number of a coloring is zero if and only if the coloring is crossing free. We now let $\gamma = \gamma_{\mathcal{D}}$ be a  disk-coloring of $\partial P$ w.r.t. disks $\mathcal{D}$, having minimum crossing number (over all disk-colorings w.r.t. $\mathcal{D}$). Assume for contradiction that the crossing number of $\gamma$ is not zero and let $r$ and $b$ be two colors of $\gamma$ which cross each other, i.e., $l_{rb} = l_{br} \geq 2$. Then, w.l.o.g., according to Lemma \ref{crossLem}, there exists a pocket  $\mathcal{P}_b$ induced by color $b$, in which all blocks of color $r$ can be colored with $b$ and the coloring remains a valid disk-coloring w.r.t. $\mathcal{D}$. We refer to the resulting disk-coloring as $\hat{\gamma}$ and by $\hat l_{ij}$ to the number of pockets induced by color $i$ of $\hat{\gamma}$ which contain blocks of color $j$ (again in $\hat{\gamma}$). Lastly we denote by $\hat{\mathcal{P}}_r$ the pocket induced by color $r$ in $\hat{\gamma}$, fully containing $\mathcal{P}_b$ as shown in Figure \ref{PolyDecomp}.

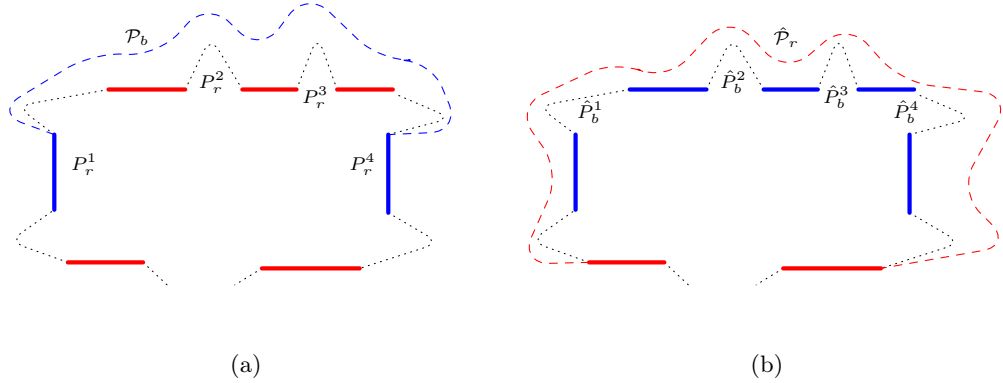
\begin{figure}[ht!]
\center
\begin{tabular}{cc}

\subfloat[]{

\begin{tikzpicture}[line cap=round,line join=round,>=triangle 45,x=1.0cm,y=1.0cm]
\clip(.3,0.25382286090537726) rectangle (6.8,5);
\draw [dotted, rounded corners=0.2cm] (1.18,1.3)-- (0.4,1.58)-- (1.0,2.0);
\draw [dotted, rounded corners=0.2cm] (1.0,3.0)-- (0.5,3.38)-- (1.72,3.6);
\draw [dotted, rounded corners=0.2cm] (2.74,3.6)-- (3.08,4.3)-- (3.5,3.6);
\draw [dotted, rounded corners=0.2cm] (4.22,3.6)-- (4.5,4.32)-- (4.76,3.6);
\draw [dotted, rounded corners=0.2cm] (5.58,3.52)-- (6.28,3.22)-- (5.44,3.0);
\draw [dotted, rounded corners=0.2cm] (5.44,1.96)-- (6.12,1.56)-- (5.06,1.22);

\draw [dotted, rounded corners=0.2cm] (2.18,1.3)-- (2.5,1);
\draw [dotted, rounded corners=0.2cm] (3.76,1.22)-- (3.4,1);

\draw [line width=1.6pt,color=red] (1.18,1.3)-- (2.18,1.3);
\draw [line width=1.6pt,color=red] (3.76,1.22)-- (5.06,1.22);
\draw [line width=1.6pt,color=red] (1.72,3.6)-- (2.74,3.6);
\draw [line width=1.6pt,color=red] (3.5,3.6)-- (4.22,3.6);
\draw [line width=1.6pt,color=red] (4.76,3.6)-- (5.5,3.6);
\draw [line width=1.6pt,color=blue] (1.0,3.0)-- (1.0,2.0);
\draw [line width=1.6pt,color=blue] (5.44,3.0)-- (5.44,1.96);
\draw [dash pattern=on 3pt off 3pt,color=blue, rounded corners=0.4cm] (1.0,3.0)-- (0.19901894358431302,3.2927164927926)-- (1.5,4.087084651455246)-- (2.3861866834341274,4.05180775242541)-- (3.038809315486088,4.774984182537042)-- (3.7443472960828026,4.087084651455246)-- (4.361693029104928,4.916091778656385)-- (5.3318077524254095,3.981253954365739)-- (6.0,4.0)-- (6.3991309738182,3.0439074739189)-- (5.44,3.0);

\begin{scriptsize}
\draw[color=black] (1.4,2.6) node {$P^1_r$};
\draw[color=black] (3.1,3.7) node {$P^2_r$};
\draw[color=black] (4.47,3.5) node {$P^3_r$};
\draw[color=black] (5.1,2.6) node {$P^4_r$};
\draw[color=black] (2.1,4.3) node {$\mathcal{P}_b$};
\end{scriptsize}
\end{tikzpicture}
}

\subfloat[]{

\definecolor{blue}{rgb}{0.0,0.0,1.0}
\definecolor{red}{rgb}{1.0,0.0,0.0}
\begin{tikzpicture}[line cap=round,line join=round,>=triangle 45,x=1.0cm,y=1.0cm]
\clip(0.3,0.25382286090537726) rectangle (7,5);
\draw [dotted, rounded corners=0.2cm] (1.18,1.3)-- (0.4,1.58)-- (1.0,2.0);
\draw [dotted, rounded corners=0.2cm] (1.0,3.0)-- (0.5,3.38)-- (1.72,3.6);
\draw [dotted, rounded corners=0.2cm] (2.74,3.6)-- (3.08,4.3)-- (3.5,3.6);
\draw [dotted, rounded corners=0.2cm] (4.22,3.6)-- (4.5,4.32)-- (4.76,3.6);
\draw [dotted, rounded corners=0.2cm] (5.58,3.52)-- (6.28,3.22)-- (5.44,3.0);
\draw [dotted, rounded corners=0.2cm] (5.44,1.96)-- (6.12,1.56)-- (5.06,1.22);

\draw [dotted, rounded corners=0.2cm] (2.18,1.3)-- (2.5,1);
\draw [dotted, rounded corners=0.2cm] (3.76,1.22)-- (3.4,1);

\draw [line width=1.6pt,color=red] (1.18,1.3)-- (2.18,1.3);
\draw [line width=1.6pt,color=red] (3.76,1.22)-- (5.06,1.22);
\draw [line width=1.6pt,color=blue] (1.72,3.6)-- (2.74,3.6);
\draw [line width=1.6pt,color=blue] (3.5,3.6)-- (4.22,3.6);
\draw [line width=1.6pt,color=blue] (4.76,3.6)-- (5.5,3.6);
\draw [line width=1.6pt,color=blue] (1.0,3.0)-- (1.0,2.0);
\draw [line width=1.6pt,color=blue] (5.44,3.0)-- (5.44,1.96);
\draw [dash pattern=on 3pt off 3pt,color=red, rounded corners=0.4cm] (1.18,1.3)-- (0.2525253639975837,1.277456268406644)-- (0.8,2.5)-- (0.19431848059835477,3.366186683434123)-- (1.5518880213785131,4)-- (2.1921637387700317,3.7)-- (3.0458646952920563,4.6)-- (3.802554179482032,3.8)-- (4.617450547071237,4.5)--  (5.323006439098885,3.7)-- (6.8,3.52128817532)-- (6.09202492651837,2.245630522097821)-- (6.809909821775527,1.5277456268406644)-- (5.06,1.22);

\begin{scriptsize}
\draw[color=black] (1.2,3.3) node {$\hat{P}^1_b$};
\draw[color=black] (3.1,3.7) node {$\hat{P}^2_b$};
\draw[color=black] (4.47,3.5) node {$\hat{P}^3_b$};
\draw[color=black] (5.4,3.3) node {$\hat{P}^4_b$};
\draw[color=black] (3.8,4.3) node {$\hat {\mathcal{P}}_r$};
\end{scriptsize}
\end{tikzpicture}
}

\end{tabular}
\caption{(a) Illustration of $\mathcal{P}_b$ in the disk-coloring $\gamma$. (b)  Illustration of $\hat {\mathcal{P}}_r$ in the disk-coloring $\hat \gamma$ .}
\label{PolyDecomp}
\end{figure}

We are going to show that $\hat{\gamma}$ has a smaller crossing number than $\gamma$, thus contradicting the assumption that $\gamma$ is the disk-coloring with minimum crossing number. For this, we extend the definition of $l_{ij}$ to parts of the polygon boundary: for a contiguous subset $\partial Q$ of $\partial P$, we denote by $l_{ij}[\partial Q]$ the number of pockets induced by color $i$ which are fully contained in $\partial Q$ and which contain blocks of color $j$.

For the rest of the proof, let $k$ be an arbitrary color of $\gamma$ (and thus also of $\hat \gamma$). Since every pocket induced by color $r$ in $\gamma$ (and in $\hat \gamma$) is either contained in $\hat {\mathcal{P}}_r$ or in $\partial P \setminus \hat {\mathcal{P}}_r$, it holds that 
\begin{align}
l_{rk} = l_{rk}[\hat {\mathcal{P}}_r] + l_{rk}[\partial P \setminus \hat {\mathcal{P}}_r] 
\hspace{10pt} \text{  and  } \hspace{10pt}  
\hat  l_{rk} = \hat l_{rk}[\hat {\mathcal{P}}_r] + \hat  l_{rk}[\partial P \setminus \hat {\mathcal{P}}_r].
\label{sameOutR}
\end{align}
Similarly, it holds that
\begin{align}
l_{bk} = l_{bk}[\mathcal{P}_b] + l_{bk}[\partial P \setminus \mathcal{P}_b]
\hspace{10pt} \text{  and  } \hspace{10pt}
\hat l_{bk} = \hat l_{rb}[\mathcal{P}_b] + \hat l_{bk}[\partial P \setminus \mathcal{P}_b].
\label{sameOutB}
\end{align}
Furthermore, since $\hat \gamma$ does not differ from $\gamma$ in $\partial P \setminus \mathcal{P}_b$, it holds that 
\begin{align}
\hat l_{bk}[\partial P \setminus \mathcal{P}_b] =  l_{bk}[ \partial P \setminus \mathcal{P}_b]
\label{sameHatB}
\end{align}
and analogously, since $(\partial P \setminus \hat {\mathcal{P}}_r)  \subseteq (\partial P \setminus \mathcal{P}_b)$, it holds that
\begin{align}
\hat l_{rk}[ \partial P \setminus \hat{ \mathcal{P}}_r] =  l_{rk}[ \partial P \setminus \hat {\mathcal{P}}_r].
\label{sameHatR}
\end{align}
Next, we are going to show that 
\begin{align}
\hat l_{rk}[\hat{ \mathcal{P}}_r] + \hat l_{bk}[{ \mathcal{P}}_b]  \leq  l_{rk}[\hat{ \mathcal{P}}_r] + l_{bk}[{ \mathcal{P}}_b].
\label{eqIneq}
\end{align}
We are going to prove this by distinguishing two cases:
1) $\hat l_{rk}[\hat{ \mathcal{P}}_r] > l_{bk}[\mathcal{P}_b]$. Since in $\hat \gamma$, by definition $\hat{ \mathcal{P}}_r$ is a single pocket induced by $r$, it follows that $\hat l_{rk}[\hat{ \mathcal{P}}_r] = 1$ and thus  $l_{bk}[\mathcal{P}_b] = 0$. Observe that $l_{bk}[\mathcal{P}_b] = 0$ means that no block of color $k$ was present in $\mathcal{P}_b$ in the $\gamma$ coloring, and this implies that $\hat l_{bk}[\mathcal{P}_b] = 0$. Furthermore, since a block of color $k$ appears inside $\hat{ \mathcal{P}}_r$ in the coloring $\hat \gamma$, a block of color $k$ appeared inside $\hat{ \mathcal{P}}_r$ in the coloring $\gamma$. Thus it holds that $ l_{rk}[\hat{ \mathcal{P}}_r] \geq 1$  which together establishes Eq.~(\ref{eqIneq}). 2) $\hat l_{rk}[\hat{ \mathcal{P}}_r] \leq l_{bk}[\mathcal{P}_b]$. We only need to show that  $\hat{l}_{bk}[\mathcal{P}_b] \leq  l_{rk}[\hat{  \mathcal{P}}_r]$. To see this, let $P^1_r, \ldots, P^t_r$ be the pockets induced by color $r$ in $\gamma$, which are contained in $\hat{ \mathcal{P}}_r$ (ordered clockwise).  Observe that since in $\hat \gamma$ each block in $\mathcal{P}_b$ which was of color $r$ in $\gamma$ gets colored with color $b$, there are again exactly $t$ such pockets $\hat  P^1_b, \ldots, \hat  P^t_b$  induced by color $b$ in $\hat \gamma$ which are fully contained in $\mathcal{P}_b$. Next, observe that for any $1 \leq p \leq t$ it holds that $\hat  P^p_b \subseteq P^p_r$. Thus if in $\hat \gamma$ a block of color $k$ is contained in a pocket $\hat P^p_b$ then this block was contained in pocket $P^p_r$ in $\gamma$. This indeed implies that $\hat{l}_{bk}[\mathcal{P}_b] \leq  l_{rk}[\hat{  \mathcal{P}}_r]$ proving Eq.~(\ref{eqIneq}) for this second case.

Using Eq.~(\ref{sameOutR}) - (\ref{eqIneq}), we obtain

\begin{align*}
\hat l_{rk} + \hat l_{bk}  
&\stackrel{ (\ref{sameOutR}),  (\ref{sameOutB})  }= \hat l_{rk}[\hat{ \mathcal{P}}_r] + \hat  l_{rk}[\partial P \setminus \hat{ \mathcal{P}}_r] + \hat l_{bk}[{ \mathcal{P}}_b] + \hat  l_{bk}[\partial P \setminus { \mathcal{P}}_b] \\
&\stackrel{ (\ref{sameHatB}),  (\ref{sameHatR})  }=  \hat l_{rk}[\hat{ \mathcal{P}}_r] +   l_{rk}[\partial P \setminus \hat{ \mathcal{P}}_r] + 
\hat l_{bk}[{ \mathcal{P}}_b] +  l_{bk}[\partial P \setminus { \mathcal{P}}_b] \\
&\stackrel{ (\ref{eqIneq}) }\leq  l_{rk}[\hat{ \mathcal{P}}_r] +   l_{rk}[\partial P \setminus \hat{ \mathcal{P}}_r] + 
l_{bk}[{ \mathcal{P}}_b] +  l_{bk}[\partial P \setminus { \mathcal{P}}_b] \\
&\stackrel{ (\ref{sameOutR}),  (\ref{sameOutB}) } = l_{rk} +  l_{bk}.
\end{align*}

Furthermore, since color $k$ was chosen arbitrarily, it holds that

\begin{align*}
\sum_{\substack{i \in \{1, \ldots, |\mathcal{D}|\}   \\ \setminus \{r,b\}}} (\hat l_{ri} + \hat l_{bi})    \leq  \sum_{\substack{i \in \{1, \ldots, |\mathcal{D}|\}   \\ \setminus \{r,b\}}} ( l_{ri} +  l_{bi}).
\end{align*}
Because we colored all blocks of color $r$ in $\mathcal{P}_b$ by color $b$, it follows that $\hat{l}_{rb} = l_{rb} -1$ and it thus indeed holds that 
 \begin{align*}
\sum_{1 \leq i < j \leq |\mathcal{D}|} \hat l_{ij} < \sum_{1 \leq i < j \leq |\mathcal{D}|} l_{ij},
\end{align*}
contradicting the assumption that $\gamma$ has the smallest crossing number.

\qed

\end{proof}

\begin{lemma}
For a {crossing-free coloring $\gamma$ using $\kappa$ colors}, let  $\Pi_\gamma$ be the set of blocks induced by  $\gamma$. If $\kappa > 1$ then $|\Pi_\gamma| \leq 2 (\kappa - 1)$.
\label{indLem}
\end{lemma}
\begin{proof}

We prove the lemma by induction on the number  of colors. For $\kappa=2$, since $\gamma$ is crossing-free it is easy to see that $|\Pi_\gamma| \leq 2$ and thus the lemma holds.  Assuming the lemma holds for $\kappa - 1$ colors, we show it also holds for $\kappa$ colors.  {For any color $i$ used by $\gamma$, let $\mathcal{B}_i \subseteq \Pi_\gamma$ be the set of blocks of color $i$}.  If for all $i$, $|\mathcal{B}_i| \leq 1$, the lemma trivially holds.  Otherwise fix $i$ to be a color with $|\mathcal{B}_i| \geq 2$ and observe that the number of pockets induced by color $i$ is $|\mathcal{B}_i|$.

\definecolor{ffqqqq}{rgb}{1.0,0.0,0.0}
\definecolor{qqwuqq}{rgb}{0.0,0.39215686274509803,0.0}
\definecolor{qqqqff}{rgb}{0.0,0.0,1.0}
\definecolor{cqcqcq}{rgb}{0.7529411764705882,0.7529411764705882,0.7529411764705882}
\definecolor{qqwuqq}{rgb}{0.0,0.39215686274509803,0.0}
\definecolor{ffqqqq}{rgb}{1.0,0.0,0.0}
\definecolor{qqqqff}{rgb}{0.0,0.0,1.0}
\definecolor{orange}{rgb}{1.0,0.5,0.0}
\definecolor{brown}{rgb}{0.5,.5,0.0}
\definecolor{yqyqyq}{rgb}{0.5019607843137255,0.5019607843137255,0.5019607843137255}
\definecolor{cqcqcq}{rgb}{0.7529411764705882,0.7529411764705882,0.7529411764705882}

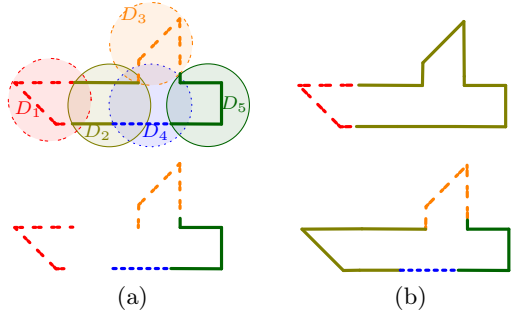
\begin{figure}[ht!]
\center
\begin{tabular}{cc}
\subfloat[]{
\begin{tikzpicture}[line cap=round,line join=round,>=triangle 45,x=1.0cm,y=1.0cm,scale=1.1]
\clip(-0.15,-1.3) rectangle (3.0,2.2);
\draw [dotted,color=qqqqff,fill=qqqqff,fill opacity=0.1] (1.6407502437867678,0.7262099843520922) circle (0.5cm);
\draw [color=qqwuqq,fill=qqwuqq,fill opacity=0.1] (2.3372699835759234,0.7288531726287284) circle (0.5cm);
\draw [dash pattern=on 1pt off 1pt on 1pt off 4pt,color=orange,fill=orange,fill opacity=0.1] (1.6481657819908717,1.4707564648272209) circle (0.5cm);
\draw [color=brown,fill=brown,fill opacity=0.1] (1.1442305039976122,0.7262099843520922) circle (0.5cm);
\draw [dash pattern=on 1pt off 1pt on 1pt off 4pt,color=ffqqqq,fill=ffqqqq,fill opacity=0.1] (0.42379009724472,0.7943597525584467) circle (0.5cm);
\draw [line width=1.2000000000000002pt,dash pattern=on 1pt off 1pt on 1pt off 4pt,color=ffqqqq] (0.5,0.5)-- (0.6967952937271886,0.5032210625048962);
\draw [line width=1.2000000000000002pt,dash pattern=on 1pt off 1pt on 1pt off 4pt,color=ffqqqq] (0.5,0.5)-- (0.0,1.0);
\draw [line width=1.2000000000000002pt,dash pattern=on 1pt off 1pt on 1pt off 4pt,color=ffqqqq] (0.0,1.0)-- (0.7246119238309283,1.00392040437221);
\draw [line width=1.2000000000000002pt,color=brown] (0.6967952937271886,0.5032210625048962)-- (1.1882224255599232,0.5032210625048962);
\draw [line width=1.2000000000000002pt,dotted,color=qqqqff] (1.1882224255599232,0.5032210625048962)-- (1.8882742831707053,0.5032210625048962);
\draw [line width=1.2000000000000002pt,dash pattern=on 1pt off 1pt on 1pt off 4pt,color=orange] (1.4942053567010596,0.9992842993549201)-- (1.499540803585664,1.2735138372423407);
\draw [line width=1.2000000000000002pt,dash pattern=on 1pt off 1pt on 1pt off 4pt,color=orange] (1.499540803585664,1.2735138372423407)-- (1.994904698568374,1.7735138372423407);
\draw [line width=1.2000000000000002pt,dash pattern=on 1pt off 1pt on 1pt off 4pt,color=orange] (1.994904698568374,1.7735138372423407)-- (1.999540803585664,1.1105508197698788);
\draw [line width=1.2000000000000002pt,color=qqwuqq] (1.999540803585664,1.1105508197698788)-- (2.0,1.0);
\draw [line width=1.2000000000000002pt,color=qqwuqq] (2.0,1.0)-- (2.5,1.0);
\draw [line width=1.2000000000000002pt,color=qqwuqq] (2.5,1.0)-- (2.5,0.5);
\draw [line width=1.2000000000000002pt,color=qqwuqq] (2.5,0.5)-- (1.8882742831707053,0.5032210625048962);
\draw [line width=1.2000000000000002pt,color=brown] (1.4942053567010596,0.9992842993549201)-- (0.7246119238309283,1.00392040437221);
\begin{scriptsize}
\draw[color=qqqqff] (1.7,0.4) node {$D_4$};
\draw[color=qqwuqq] (2.652873888891698,0.7410756695859958) node {$D_5$};
\draw[color=orange] (1.4436826803576182,1.8148374627642583) node {$D_3$};
\draw[color=brown] (1,0.4) node {$D_2$};
\draw[color=ffqqqq] (0.16677676414562995,0.68303449157636) node {$D_1$};


\end{scriptsize}

\draw [line width=1.2000000000000002pt,dash pattern=on 1pt off 1pt on 1pt off 4pt,color=ffqqqq] (0.5,0.5 - 1.75 )-- (0.6967952937271886,0.5032210625048962   - 1.75 );
\draw [line width=1.2000000000000002pt,dash pattern=on 1pt off 1pt on 1pt off 4pt,color=ffqqqq] (0.5,0.5 - 1.75 )-- (0.0,1.0 - 1.75 );
\draw [line width=1.2000000000000002pt,dash pattern=on 1pt off 1pt on 1pt off 4pt,color=ffqqqq] (0.0,1.0 - 1.75 )-- (0.7246119238309283,1.00392040437221 - 1.75 );

\draw [line width=1.2000000000000002pt,dotted,color=qqqqff] (1.1882224255599232,0.5032210625048962 - 1.75 )-- (1.8882742831707053,0.5032210625048962 - 1.75 );
\draw [line width=1.2000000000000002pt,dash pattern=on 1pt off 1pt on 1pt off 4pt,color=orange] (1.4942053567010596,0.9992842993549201 - 1.75 )-- (1.499540803585664,1.2735138372423407 - 1.75 );
\draw [line width=1.2000000000000002pt,dash pattern=on 1pt off 1pt on 1pt off 4pt,color=orange] (1.499540803585664,1.2735138372423407 - 1.75 )-- (1.994904698568374,1.7735138372423407 - 1.75 );
\draw [line width=1.2000000000000002pt,dash pattern=on 1pt off 1pt on 1pt off 4pt,color=orange] (1.994904698568374,1.7735138372423407 - 1.75 )-- (1.999540803585664,1.1105508197698788 - 1.75 );
\draw [line width=1.2000000000000002pt,color=qqwuqq] (1.999540803585664,1.1105508197698788 - 1.75 )-- (2.0,1.0 - 1.75 );
\draw [line width=1.2000000000000002pt,color=qqwuqq] (2.0,1.0 - 1.75 )-- (2.5,1.0 - 1.75 );
\draw [line width=1.2000000000000002pt,color=qqwuqq] (2.5,1.0 - 1.75 )-- (2.5,0.5 - 1.75 );
\draw [line width=1.2000000000000002pt,color=qqwuqq] (2.5,0.5 - 1.75 )-- (1.8882742831707053,0.5032210625048962 - 1.75 );

\end{tikzpicture}
}
& \subfloat[] {
\begin{tikzpicture}[line cap=round,line join=round,>=triangle 45,x=1.0cm,y=1.0cm,scale=1.1]
\clip(0.0,.25) rectangle (3.0,3.4);
\draw [line width=1.2pt,color=brown] (0.6859854505161606,0.2804467947854956)-- (0.8827807442433493,0.2836678572903919);
\draw [line width=1.2pt,color=brown] (0.6859854505161606,0.2804467947854956)-- (0.18598545051616067,0.7804467947854961);
\draw [line width=1.2pt,color=brown] (0.18598545051616067,0.7804467947854961)-- (0.9105973743470896,0.7843671991577068);
\draw [line width=1.2pt,color=brown] (0.8827807442433493,0.2836678572903919)-- (1.3742078760760807,0.2836678572903919);
\draw [line width=1.2000000000000002pt,dotted,color=qqqqff] (1.3742078760760807,0.2836678572903919)-- (2.0742597336868616,0.2836678572903919);
\draw [line width=1.2000000000000002pt,dash pattern=on 1pt off 1pt on 2pt off 4pt,color=orange] (1.680190807217215,0.7797310941404167)-- (1.6855262541018194,1.0539606320278376);
\draw [line width=1.2000000000000002pt,dash pattern=on 1pt off 1pt on 2pt off 4pt,color=orange] (1.6855262541018194,1.0539606320278376)-- (2.180890149084529,1.5539606320278372);
\draw [line width=1.2000000000000002pt,dash pattern=on 1pt off 1pt on 2pt off 4pt,color=orange] (2.180890149084529,1.5539606320278372)-- (2.1855262541018194,0.8909976145553751);
\draw [line width=1.2000000000000002pt,color=qqwuqq] (2.1855262541018194,0.8909976145553751)-- (2.1859854505161596,0.7804467947854961);
\draw [line width=1.2000000000000002pt,color=qqwuqq] (2.1859854505161596,0.7804467947854961)-- (2.6859854505161604,0.7804467947854961);
\draw [line width=1.2000000000000002pt,color=qqwuqq] (2.6859854505161604,0.7804467947854961)-- (2.6859854505161604,0.2804467947854956);
\draw [line width=1.2000000000000002pt,color=qqwuqq] (2.6859854505161604,0.2804467947854956)-- (2.0742597336868616,0.2836678572903919);
\draw [line width=1.2pt,color=brown] (1.680190807217215,0.7797310941404167)-- (0.9105973743470896,0.7843671991577068);
\draw [line width=1.2000000000000002pt,dash pattern=on 1pt off 1pt on 2pt off 4pt,color=ffqqqq] (0.6473590465507528,2.018634973229065)-- (0.8441543402779409,2.021856035733961);
\draw [line width=1.2000000000000002pt,dash pattern=on 1pt off 1pt on 2pt off 4pt,color=ffqqqq] (0.6473590465507528,2.018634973229065)-- (0.14735904655075305,2.518634973229065);
\draw [line width=1.2000000000000002pt,dash pattern=on 1pt off 1pt on 2pt off 4pt,color=ffqqqq] (0.14735904655075305,2.518634973229065)-- (0.8719709703816804,2.522555377601276);
\draw [line width=1.2pt,color=brown] (0.8441543402779409,2.021856035733961)-- (1.3355814721106687,2.021856035733961);
\draw [line width=1.2pt,color=brown] (1.3355814721106687,2.021856035733961)-- (2.0356333297214486,2.021856035733961);
\draw [line width=1.2pt,color=brown] (1.641564403251802,2.5179192725839856)-- (1.6468998501364065,2.7921488104714065);
\draw [line width=1.2pt,color=brown] (1.6468998501364065,2.7921488104714065)-- (2.142263745119116,3.292148810471407);
\draw [line width=1.2pt,color=brown] (2.142263745119116,3.292148810471407)-- (2.1468998501364065,2.6291857929989444);
\draw [line width=1.2pt,color=brown] (2.1468998501364065,2.6291857929989444)-- (2.1473590465507466,2.518634973229065);
\draw [line width=1.2pt,color=brown] (2.1473590465507466,2.518634973229065)-- (2.6473590465507475,2.518634973229065);
\draw [line width=1.2pt,color=brown] (2.6473590465507475,2.518634973229065)-- (2.6473590465507475,2.018634973229065);
\draw [line width=1.2pt,color=brown] (2.6473590465507475,2.018634973229065)-- (2.0356333297214486,2.021856035733961);
\draw [line width=1.2pt,color=brown] (1.641564403251802,2.5179192725839856)-- (0.8719709703816804,2.522555377601276);
\begin{scriptsize}
\end{scriptsize}
\end{tikzpicture}
}
\end{tabular}
\caption{(a) A disk-coloring example w.r.t. disks $D_1, \ldots, D_5$; in the lower part the two pockets induced by color $2$ are shown.  (b) shows the two polygon colorings in the induction step for color $2$ in the proof of Lemma~\ref{indLem}.}
\label{colorExFigProof}
\end{figure}

 Let $\mathcal{P}_1, ..., \mathcal{P}_{|\mathcal{B}_i|}$ be the pockets induced by color $i$.  For each such pocket $\mathcal{P}_j$ we create a new coloring $\gamma_j$ of $\partial P$, with
 $$
 \gamma_j(x) =
  \begin{cases} 
      \hfill \gamma(x) \hfill & \text{ if $x \in \mathcal{P}_j$} \\
      \hfill i  \hfill & \text{ otherwise, }  \\
  \end{cases}
$$
as illustrated in Fig. \ref{colorExFigProof}(b).

{Since $\gamma$ is crossing-free it is easy to see that  for any pocket $\mathcal{P}_j$, the coloring $\gamma_j$ is also a crossing-free.  Denoting the number of colors of $\gamma_j$ by $\kappa_j$, it holds that $1 <  \kappa_j < \kappa$, for all $1 \leq j \leq |\mathcal{B}_i|$}.  Letting $\Pi_{\gamma_j}$ be the {set of blocks} induced by the coloring $\gamma_j$, by induction hypothesis $|\Pi_{\gamma_j}| \leq 2(\kappa_j-1)$.  

Observe that each $\Pi_{\gamma_j}$ contains exactly one block not in $\Pi_\gamma$.  Also, the blocks in $\mathcal{B}_i$ are exactly the blocks not appearing in any of the $\Pi_{\gamma_j}$.  Therefore, since the number of pockets induced by color $i$ equals the number of blocks in {$\mathcal{B}_i$},
it holds that $|\Pi_\gamma| =  \sum_{j=1}^{|\mathcal{B}_i|}  |\Pi_{\gamma_j}|$. Thus we obtain

\begin{equation}
|\Pi_\gamma| =  \sum_{j=1}^{|\mathcal{B}_i|}  |\Pi_{\gamma_j}| \leq 2\sum_{j=1}^{|\mathcal{B}_i|}  (\kappa_j - 1)
\end{equation}
 
\noindent and because each of the colorings $\gamma_1,  \ldots, \gamma_{|\mathcal{B}_i|}$ is crossing-free,

\begin{equation}
1 + \sum_{j=1}^{|\mathcal{B}_i|} (\kappa_j  - 1)  =  \kappa, 
\end{equation}
where $(\kappa_j - 1)$ is the number of colors the coloring {$\gamma$} (and also $\gamma_j$) uses for the pocket  $\mathcal{P}_j$.  {The addition of $1$ on the left hand side of (3) attributes for  color $i$}, which was not counted in any of the pockets.
Plugging $(3)$ into $(2)$, the lemma follows. 
\qed

\end{proof}



\begin{theorem}
The number of disk centers placed by \textsc{ContiguousGreedy} is at most $2|OPT| - 1$.
\end{theorem}
\begin{proof}

If $|OPT| = 1$ then, by its greedy nature, \textsc{ContiguousGreedy} also uses only one disk. 
If $|OPT| > 1$, let $\gamma_{{OPT}}$ be a crossing free disk-coloring of $\partial P$ w.r.t. ${OPT}$, whose existence is guaranteed by Lemma \ref{crossFreeExistsLem}. We let $(B_1, B_2, ..., B_m)$ be the collection of blocks induced by $\gamma_{{OPT}}$ ordered as they appear on $\partial P$  in clockwise order, with $B_1$ the block containing $v_1$. We split $B_1$ at $v_1$ into two blocks $B_l$ and $B_r$, with $B_l$ being the portion of $B_1$ counterclockwise from $v_1$, and $B_r = B_1 \setminus B_l$.

Now observe that by the greedy nature, every disk $D$ computed by \textsc{ContiguousGreedy} extends $\Gamma $ so that $\Gamma \cup D$ fully covers at least one new block in the sequence $(B_r, B_2, ..., B_m, B_l)$.  Therefore, after computing at most $m+1$ disks, $\Gamma = \partial P$. By Lemma \ref{indLem}, it holds that $m\leq 2(|OPT| - 1)$ and the theorem follows.  
\qed
\end{proof}

\subsection{Tightness of Analysis}
The analysis for the $2$-approximation ratio of \textsc{ContiguousGreedy} is almost tight, even for convex polygons, as can be seen by a rectangle of length $n$ and height $\epsilon > 0$. It can be covered with $n/(2\sqrt{1-\epsilon^2/4})$ many geodesic unit disks (by centering them on the median line at height $\epsilon/2$). On the other hand, \textsc{ContiguousGreedy} centers disks in steps of $2$ on the boundary, thus after finishing one side of the rectangle, each disk introduced a small uncovered hole on the other side. \textsc{ContiguousGreedy} covers those holes by placing another $n/2$ disks contiguously on the other side of the polygon, resulting in a total of $n$ disks needed.\\

\vspace{-30pt}

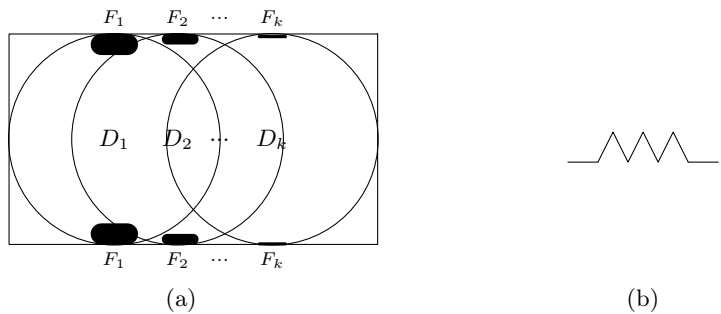
\begin{figure}
\center
\begin{tabular}{cc}
\subfloat[]{
\definecolor{xdxdff}{rgb}{0.49,0.49,1}
\definecolor{black}{rgb}{0,0,1}
\definecolor{cqcqcq}{rgb}{0.75,0.75,0.75}
\begin{tikzpicture}[scale=0.7, line cap=round,line join=round,>=triangle 45,x=1.0cm,y=1.0cm]
\draw (6,5)-- (-1,5)-- (-1,1)-- (6,1) -- (6,5);

\draw [line width=8pt] (0.75,4.8)-- (1.25,4.8);
\draw [line width=8pt] (0.75,1.2)-- (1.25,1.2);

\draw [line width=4pt] (2,4.9)-- (2.5,4.9);
\draw [line width=4pt] (2,1.1)-- (2.5,1.1);


\draw [line width=1.2pt] (3.75,4.95)-- (4.25,4.95);
\draw [line width=1.2pt] (3.75,1.01)-- (4.25,1.01);

\begin{scriptsize}
\node at (1, 5.3) {$F_1$};
\node at (1, 0.7) {$F_1$};

\node at (3, 5.3) {$...$};
\node at (3, .7) {$...$};

\node at (4, 5.3) {$F_k$};
\node at (4, .7) {$F_k$};

\node at (2.2, 5.3) {$F_2$};
\node at (2.2, 0.7) {$F_2$};

\end{scriptsize}

\draw[] (1,2.999) circle(2.01);
\node at (1, 3) {$D_1$};

\draw[] (2.2,2.999) circle(2.01);
\node at (2.2, 3) {$D_2$};

\node at (3, 2.999) {$...$};

\draw[] (4,2.999) circle(2.01);
\node at (4, 3) {$D_k$};

\end{tikzpicture}
}
& \subfloat[] {

\begin{tikzpicture}[line cap=round,line join=round,>=triangle 45,x=1.0cm,y=1.0cm]
\clip(-3.4,-1.5) rectangle (3,1.2);
\draw (-0.8,0)-- (-0.6,0.4)-- (-0.4,0)-- (-0.2,0.4)-- (0,0)-- (0.2,0.4)-- (0.4,0);
\draw (0.4,0)-- (.8,0);
\draw (-0.8,0)-- (-1.2,0);
\end{tikzpicture}

}
\end{tabular}
\caption{(a) Illustration of the polygon containing foldings $F_1, \ldots, F_k$ on the boundary. A global greedy algorithm starts covering the two $F_1$ foldings on opposite sides by the disk $D_1$, the two $F_2$ foldings by a disk $D_2$ and so on, while $OPT$ still only uses a constant number of disks to cover $\partial P$. (b) Illustration of a folding. }
\label{globalGreedy}
\end{figure}

Another natural greedy approach is to cover the largest amount of uncovered boundary at each step.  This algorithm results in an approximation ratio of $\Omega(\log n)$, i.e., it is unbounded with respect to $|OPT|$. An example where this greedy rule performs badly is illustrated in Fig.~\ref{globalGreedy}(a). The parts of the boundary denoted by $F_1, \ldots, F_k$ are dense {\em foldings} as shown in Fig.~\ref{globalGreedy}(b) where the boundary length of $F_1$ is twice that of $F_2$, four times that of $F_3$, and so on.  The global greedy algorithm first covers the two $F_1$ sections on opposite sides of the boundary (illustrated by $D_1$ in Fig.~\ref{globalGreedy}(a)), then the two $F_2$ sections continuing in this way until the two $F_k$ sections are covered, thereby having used $k$ disks to cover the foldings, (plus some constant number of disks to cover the rest of $\partial P$). Notice that when the height of the polygon is arbitrary close to $2$, the number of foldings can be made arbitrary large, while $OPT$ only uses a constant number of disks to cover $\partial P$.\\

It is worth noting that it is crucial that  \textsc{ContiguousGreedy} $\emph{exactly}$ computes the maximum extension of the contiguous boundary covered by a single geodesic unit disks in each iteration. Only approximately (even with $\epsilon$ precision) extending the contiguously covered part by a single geodesic unit disks results in an approximation factor of at least $4$ (instead of $2$). 
To see this, we refer to  Fig.~\ref{counterRef}, where $c_1$ is the endpoint of the $\epsilon$-approximate contiguous greedy extension in the first step and $c^*_1$ is the corresponding exact endpoint (obtained from \textsc{ContiguousGreedy}). The $\epsilon$-approximate algorithm continues by centering a disks at $D_2$ which covers the boundary from $c_1$  up to $c_2$. At this point, an exact extension could cover the boundary from $c_2$ up to $c^*_2$. However, the approximate algorithm may only cover up to $c_3$, by, for example, centering the third disk at $D_3$.  \textsc{ContiguousGreedy} covers up to $c^*_2$ using only two disks. Copying the polygon-section between $c^*_1$ and $c^*_2$, shows that an $\epsilon$-approximate algorithm performs at least twice as bad as \textsc{ContiguousGreedy}.

\begin{figure}
\center
\begin{tikzpicture}[line cap=round,line join=round,>=triangle 45,x=1.0cm,y=1.0cm, scale = 1.4]
\clip(-2.5,0.25) rectangle (5.25,3.25);
\draw [line width=3.2pt] (-2.,1.)-- (3.,1.);
\draw [color = white, line width=2.2pt] (-2.,1.)-- (3.,1.);

\draw [line width=3.2pt] (3.,1.)-- (3.,3.);
\draw [color = white, line width=2.2pt] (3.,1.)-- (3.,3.);

\draw [dotted] (0.6,1.)-- (0.6,0.6);
\draw [dotted] (1.,1.)-- (1.,0.6);
\draw [dotted] (3.,1.)-- (3.,0.6);
\draw [dotted] (3.,3.03)-- (3.4,3.03);

\draw [<->, dotted] (-2.,1.4)-- (1.,1.4);
\draw [dotted] (-2.03,1.05)-- (-2.03,1.4);
\draw [dotted] (1,1.)-- (1,1.4);

\draw [line width=3.2pt] (3,1.)-- (5.,1.);
\draw [color = white, line width=2.2pt] (3,1.)-- (5.,1.);

\draw [dotted] (-2.06, 1.04)-- (-2.3, 1.04);
\draw [dotted] (-2.06, 0.97)-- (-2.3, 0.97);

\draw[color=black] (-2.4,1) node {$\delta $};

\draw [<->, dotted] (3.4,1)-- (3.4,3.03);
\draw [<->, dotted] (0.6,0.6)-- (1.,0.6);
\draw [<->, dotted] (1.,0.6)-- (3.,0.6);


\draw [<->, dotted] (2.4,2.6)-- (2.4,3.03);
\draw[color=black] (2.,2.8) node {$\epsilon - \delta $};
\draw [dotted] (3.,3.03)-- (2.4,3.03);
\draw [dotted] (3,2.6)-- (2.4,2.6);


\begin{scriptsize}


\draw [fill=black] (1,1.03) circle (1.pt);
\draw[color=black] (1.25157995835,1.165648189885) node {$c^*_1$};

\draw [fill=black] (2.4,1.03) circle (1.pt);
\draw[color=black] (2.285993056,1.15648189885) node {$D_2$};

\fill [white] (2.94,1.03) rectangle (6.017,0.978);

\fill [white] (2.98,1.04) rectangle (3.02,0.978);

\draw [fill=black] (-0.685993056,1.03) circle (1pt);
\draw[color=black] (-0.6885993056,1.15648189885) node {$D^1$};

\draw [fill=black] (3.03,1.4) circle (1.pt);
\draw[color=black] (3.2085993056,1.45648189885) node {$D_3$};

\draw [fill=black] (.6,1.03) circle (1.pt);
\draw[color=black] (0.680025181185,1.15648189885) node {$c_1$};

\draw [fill=black] (2.97,2.6) circle (1.pt);
\draw[color=black] (2.75,2.5) node {$c_2$};

\draw [fill=black] (4.4,1.03) circle (1.pt);
\draw[color=black] (4.480025181185,1.15648189885) node {$c_3$};

\draw [fill=black] (5,1.03) circle (1.pt);
\draw[color=black] (5.180025181185,1.15648189885) node {$c^*_2$};

\draw[color=black] (.78,0.82) node {$\epsilon$};
\draw[color=black] (0,1.28) node {$2$};
\draw[color=black] (2,0.72) node {$1$};
\draw[color=black] (3.8,2) node {$1-\delta$};

\end{scriptsize}
\end{tikzpicture}
\caption{Illustration of the $\delta$-thin polygon where an $\epsilon$-approximate contiguous extension algorithm results in an approximation ratio larger than $2$.}
\label{counterRef}
\end{figure}

\section{Covering Large Perimeters}
\label{refAna}

In this section we show that if the polygon perimeter $L$ is significantly larger than $n$, i.e., $L \geq n^{1 + \delta}$, with $\delta > 0$, a simple linear time algorithm achieves an approximation ratio which goes to one as $L/n$ goes to infinity. For this, we decompose $\partial P$ into long and short portions, based on the length of the corresponding \emph{medial axis}. The medial axis is the set of points in $P$ which have more than one closest point on $\partial P$. It forms a tree whose edges are either line segments or parabolic arcs and it can be computed in linear time \cite{medialAxis}. For a line segment edge, the closest points to the boundary are a subset of two polygon edges; for a parabolic edge, the closest boundary points are a polygon vertex and a subset of a polygon edge.  The idea of the algorithm is to identify long edges of the medial axis (of length at least some constant $c >2$), and to cover the corresponding polygon boundary section (referred to as {\em corridors}) almost optimally using only a constant number of disks more than $OPT$ uses to cover the corridor. It is easy to see that each corridor stemming from a parabolic arc can be covered with at most two more disks than $OPT$ uses, by centering disks at distance $2$ from each other on the corresponding polygon boundary segment and one disk on the corresponding polygon vertex.  Each corridor consisting of a pair of polygon boundary segments can be covered by greedily centering disks on the corresponding medial axis as long as each disk contains corridor portions of length more than two; if the length becomes two or less, greedily center the disks on  corridor segments in steps of two.  Observe that also in this case, the number of disks needed to cover a corridor is at most two more than $OPT$ uses and their centers can be computed in time linear in their number. This holds since there is at most one point where the covering changes from centering disks on the medial axis to centering disks on $\partial P$. The rest of the polygon, i.e., the short portions,  can be covered greedily by centering $O(n)$ disks on $\partial P$.

Let $\mathcal{D}$ be the set of all disks placed by the algorithm, $\mathcal{D}_L \subseteq \mathcal{D}$ the disks covering the corridors and $\mathcal{D}_S \subseteq \mathcal{D}$ the $O(n)$ disks covering the short portion of $\partial P$.  
Since the number of edges in the medial axis is $O(n)$ (see \cite{medialAxis}) and the procedure for covering the long corridors uses at most two more disks than $OPT$ for each corridor, $|\mathcal{D}_L| \leq |OPT| + O(n)$. It therefore holds that $|\mathcal{D}| = |\mathcal{D}_L| + |\mathcal{D}_S| \leq |OPT| + O(n)$. It is easy to see that the disks of $OPT$ which contain a polygon vertex cover at most an $O(n)$ portion of $\partial P$ implying that $|OPT| =  \Omega(L) $.  Therefore, the approximation ratio can be written as
$$\frac{|\mathcal{D}|}{|OPT|} \leq 1 + \frac{O(n)}{|OPT|} = 1 + \frac{O(n)}{\Omega(L)} = 1 + O\left( \frac{n}{L} \right ) = 1 + o\left( 1 \right ).$$

\section{Acknowledgments}
\vspace{-8pt}
We thank Alon Efrat for his idea of looking at large perimeter polygons. We further thank the anonymous referees for their review of a previous version of this manuscript.

\vspace{-8pt}

\bibliographystyle{abbrv}	
\bibliography{refDrop}

\end{document}